\documentclass[journal]{IEEEtran}
\usepackage{amsthm}
\usepackage{amssymb}
\usepackage[active]{srcltx}
\usepackage{amsfonts}
\usepackage{graphics}
\usepackage{eurosym}
\usepackage{pstricks}
\usepackage{shortvrb}
\usepackage{graphicx}
\usepackage{subfigure}
\usepackage{color}
\usepackage{multirow}
\usepackage{float}
\usepackage{ifpdf}
\usepackage{multicol}
\usepackage{epstopdf}
\usepackage{lineno}
\usepackage{cite}
\usepackage{array}
\usepackage{supertabular}
\usepackage{xfrac}
\usepackage{amsmath}
\usepackage{amssymb}
\usepackage{framed}
\usepackage{enumitem}
\usepackage[boxruled,vlined]{algorithm2e}
\SetKwRepeat{Do}{do}{while}
\usepackage{tikz}
\usetikzlibrary{shapes,arrows}

\makeatletter

\newtheorem{theorem}{Theorem}

\usepackage{url}
\usepackage{color}
\usepackage{soul}
\soulregister\ref7
\soulregister\eqref7
\usepackage{alphalph}
\usepackage{etoolbox}
\usepackage{tabularx}
\usepackage{xcolor}

\begin{document}
%
\title{Real-time Control of Battery Energy Storage Systems to Provide Ancillary Services Considering Voltage-Dependent Capability of DC-AC Converters}

\author{%
\IEEEauthorblockN{Zhao Yuan,~\IEEEmembership{Member,~IEEE}, Antonio Zecchino,~\IEEEmembership{Member,~IEEE},  Rachid Cherkaoui,~\IEEEmembership{Senior Member,~IEEE}, \\Mario Paolone~\IEEEmembership{Senior Member,~IEEE}}

\thanks{The authors are with the Distributed Electrical Systems Laboratory, École Polytechnique Fédérale de Lausanne, Switzerland (EPFL), e-mail: yuanzhao@kth.se, \{zhao.yuan, antonio.zecchino, rachid.cherkaoui, mario.paolone\}@epfl.ch.}
\thanks{This work has received funding from the European Union’s Horizon 2020 research and innovation programme under grant agreement n°773406.}
}

\markboth{IEEE Trans. On Smart Grid.}{}
\maketitle
\begin{abstract}
Frequency response and voltage support are vital ancillary services for power grids. In this paper, we design and experimentally validate a real-time control framework for battery energy storage systems (BESSs) to provide ancillary services to power grids. The objective of the control system is to utilize the full capability of the BESSs to provide ancillary services. We take the voltage-dependent capability curve of the DC-AC converter and the security requirements of BESSs as constraints of the control system. The initial power set-points are obtained based on the droop control approach. To guarantee the feasibility of the power set-points with respect to both the converter capability and BESS security constraints, the final power set-points calculation is formulated as a nonconvex optimization problem. A convex {and computationally efficient} reformulation of the original control problem is then proposed. We prove that the proposed convex optimization gives the global optimal solution to the original nonconvex problem. We improve the computational performance of this algorithm by discretizing the feasible region of the optimization model. We achieve a 100 ms update time of the controller setpoint computation in the experimental validation of the utility-scale 720 kVA / 560 kWh BESS on the EPFL campus.
\end{abstract}
\section*{Nomenclature}
\setlength{\parindent}{0.5em}
\textbf{Index:}\\
\begin{tabular}{l l}
$t$ & Time step. 
\end{tabular}

\textbf{Variables:} \\
\begin{tabular}{l l}
$P^{AC}_{t},Q^{AC}_{t}$ & Active and reactive power set-points \\
&on the AC-bus of the BESS.\\
$P^{DC}_t$ & Active power on the DC-bus of BESSs. \\
$v^{DC}_t,i^{DC}_t$ & Voltage and current on the BESS DC-bus. \\
$SoC_{t}$ & Battery {state-of-charge.}\\
$obj,obj^{M},obj^{S}$ & Objective function of the original,\\
&modified and static optimization problem.\\
\end{tabular}

\textbf{Parameters:} \\
\begin{tabular}{l l}
$\Delta t, \theta^{AC}_t$ &Time interval and grid voltage phase angle.\\
$P^{AC}_{0,t},Q^{AC}_{0,t}$ &Initial active and reactive power set-points\\
&on the AC-bus of the BESS.\\
$f_{t},v^{AC}_{t}, i^{AC}_{t}$ &Grid frequency, voltage and current\\
&measurements.\\
\end{tabular}

\begin{tabular}{l l}
$f^{nom}$ & Nominal grid frequency.\\
$v^{ACnom}$ & Nominal grid voltage.\\
$\Delta f_{t},\Delta v^{AC}_{t}$ & Deviations of frequency and voltage with\\
&respect to the nominal values.\\
$\alpha,\beta$ & Frequency and voltage droop coefficients.\\
$S^{ACmax}$ & Capacity of BESSs DC-AC converter.\\
$SoC^{min}$ & Lower bound of $SoC$.\\
$SoC^{max}$ & Upper bound of $SoC$.\\
$v^{DCmin}$ & Lower bound of $v^{DC}_t$.\\
$v^{DCmax}$ & Upper bound of $v^{DC}_t$.\\
$v^{max}_s$ & Upper bounds of $v_s$.\\
$i^{DCmax}$ & Upper bound of $i^{DC}_t$.\\
$C^{max}$ &Charge capacity of the battery.\\
$\eta$ & Energy efficiency of the DC-AC converter.\\
$v_{C1},v_{C2}$ &Internal voltages of the battery equivalent\\
&circuit model.\\
$v_{C3},v_s$ &Internal voltages of the battery equivalent\\
&circuit model.\\
$E$ &Open-circuit voltage of the battery equivalent\\
&circuit model.\\
$R_1,R_2$ &Internal resistances of the battery equivalent\\
&circuit model.\\
$R_3,R_s$ &Internal resistances of the battery equivalent\\
& circuit model.\\
$C_1,C_2,C_3$ &Internal capacitances of the battery equivalent\\
&circuit model.\\
$a, b$ &Parameters to estimate $E$.\\
$\xi$ &Penalty parameter.\\
\end{tabular}

{\textbf{Metrics:}}\\
\begin{tabular}{l l}
TDE &Total discharged energy. \\
TCE &Total charged energy. \\
TSE &Total sustained energy. 
\end{tabular}
\section{Introduction}
Power system dynamics in traditional grids has been largely dominated by the responses of synchronous generators \cite{kundur1994power}. With the ever increasing connection of power electronics based generation, new schemes of frequency and voltage controls are in demand \cite{mario2020pscc,Surv_Fre_volt1,Surv_Fre_volt2}. The analysis of Australia blackouts during September 2016 and August 2018 show very large rate-of-change-of-frequency (RoCoF) calling for the provision of fast frequency control services \cite{AEM2017blackout,AEM2019blackout}. In this regard, battery energy storage systems (BESSs) operating in grid-following or grid-forming mode can provide various benefits to enhance the stability of power systems \cite{mario2020pscc,GridformInveACri,RevGriConConvShor,RolPowEleFutLow,Opt_Bat_Ene}. The economic analysis {in} \cite{Opt_LiFe_UK,BOZORG2018270}, based on the estimated battery lifetime and UK frequency regulation market, shows that deploying BESS is profitable in the UK market. The lowest tender price  for BESS firm frequency control provider is 17.4 \pounds /MW/h which is lower than the tender prices for many traditional frequency control providers.  
Recent advances in the control of BESSs to provide grid ancillary services can be categorized into four areas: (1) simultaneous provision of multiple services such as energy, frequency control and voltage support \cite{Cont_Bat_Simu_Mult,Achi_Disp_Feed,zecchino2019optimal}, (2) optimal operation considering various BESS constraints and the solution algorithms \cite{Flex_Opt_Oper,Resi_Ener_Stor_Bid,Hyb_Con_Net_Bat,zecchino2019optimal}, (3) converter control techniques with improved performance for weak or low-inertia power grid \cite{GridformInveACri,RevGriConConvShor,RolPowEleFutLow}. (4) integrated operation with renewable energy resources \cite{Cord_Cont_Stra,BESSEnaDisSol,IncorBESSMultMWPV}. {{In}} the following, we specifically focus on the literature in the area (2) of BESSs control as it is the most relevant with respect to the technical content of this paper.
Authors in \cite{Cont_Bat_Simu_Mult} propose a control framework for BESSs to provide both energy and frequency control services. In this framework, an energy output schedule of the BESSs is firstly optimized and, then, the frequency control is superimposed. Based on a model predictive control (MPC), {reference} \cite{Achi_Disp_Feed} achieves dispatchability of the distribution feeder in a 5 min time resolution considering the BESS operation constraints. 
In \cite{Flex_Opt_Oper}, the daily operation of BESS is formulated {as} a mixed integer nonlinear programming problem (MINLP) to optimize the time duration of charge/discharge. The formulated MINLP is solved with a two-stage approach in which the integer variables are solved in the first-stage and a nonlinear programming problem (NLP) is solved in the second-stage \cite{Flex_Opt_Oper}. Reference \cite{Resi_Ener_Stor_Bid} proposes an energy storage management based on a stochastic optimization problem to minimize the net system cost considering bidirectional energy flow with the grid. The stochastic optimization problem is reformulated and solved through the Lyapunov optimization approach which can deal with the uncertainties from renewable energy generation, power load and electricity price \cite{Resi_Ener_Stor_Bid}. The authors of \cite{Resi_Ener_Stor_Bid} prove the existence of the upper bound of the {so called} Lyapunov drift and use this upper bound to design the control algorithm to minimize a proposed drift-plus-cost metric. 

{The aforementioned references are focusing on BESS optimal control on relatively long time horizons $>$ 1 s. The consideration of very short time horizon, especially for control response around 100 ms, is not well discussed in the exisiting literature.} To control the battery state-of-charge (SoC), {reference} \cite{Hyb_Con_Net_Bat} proposes a hybrid control methodology and proves the convergence and optimality of the proposed algorithm.
Integrated operation of BESS with wind power plant can ensure the commitment of wind power plant to provide primary and secondary frequency control services \cite{Cord_Cont_Stra}. A state-machine based controller in \cite{Cord_Cont_Stra} considers the operational constraints of wind power plan and the BESS. The battery {SoC} is kept at the preferred level by the adaptive feedback control \cite{Cord_Cont_Stra}. {In all the above listed references}, the {voltage-dependent} capability curve of {the DC-AC} converter is not considered. 
The capability of the DC-AC converter is determined by the physical constraints (current or thermal limits of switching devices e.g. IGBTs and connectors) and operational parameters (voltages on the DC-bus and AC side). Due to the variation of these operational parameters, the actual capability of the DC-AC converter is voltage-dependent. Various benefits of considering the actual full capability of power generation resources are shown in \cite{Ext_Reac_Cap_DFIG,Exa_Thr_ACOPF_Cap,Effe_Reac_Pow_Max,johansson2004power}. Capability curve of doubly fed induction generator (DFIG) wind turbine is considered in \cite{Ext_Reac_Cap_DFIG} {for generation dispatch and voltage control. The simulation results of the 26-bus power network show a significant reduction of power losses and improved post-fault voltage profiles} due to the full commitment of DFIG capability.\cite{Ext_Reac_Cap_DFIG} \cite{Effe_Reac_Pow_Max} approximates the generator capability curve by quadratic functions considering the terminal voltage, field and armature current limits. Using the generator capability curve in the optimal power flow (OPF) model,  the results for the CIGRE 32-bus and 1211-bus power networks show that representing generator capability is important when stressed conditions are studied \cite{Effe_Reac_Pow_Max}. As a matter of fact, there is a research gap in considering the voltage-dependent capability of DC-AC converters in the real-time control of BESSs.

In summary, two research gaps are identified from the exisiting literature about the real-time control of BESSs: (1) how to formulate and solve a suitable optimal control algorithm in order to update the BESS control loop to provide ancillary services within a very low latency (e.g. 100 ms). and  (2) how to embed the voltage-dependent capability of the DC-AC converter, which evolves according to the DC-bus voltage and AC-side voltage. To tackle these research gaps, the contributions of this paper are: 
\begin{itemize}
\item formulation of the BESS real-time control as a nonlinear optimization problem and its convexification considering the voltage-dependent capability curve of the DC-AC converter;
\item analytical proof of the solution optimality of the convexified optimization problem with respect to the original one;
\item developement of two algorithms to efficiently solve the optimization problem and achieving computation times compatible with the BESS realtime control.
\end{itemize}
There are several benefits of using the proposed algorithms to operate the BESS:
\begin{itemize}
\item economic efficiency for the BESS operator to avoid expensive global optimization solvers to find the global optimal solution of the initial nonconvex optimization problem;
\item provision of fast and optimal BESS frequency regulation and voltage support to the power grid;
\item continuity of the BESS operation for ancillary grid service delivery.
\end{itemize}
The rest of this paper is organized as follows. Section II formulates the optimization models and illustrates the solution algorithms for the BESSs real-time control. Section III presents the experiments for frequency control and voltage support. Section IV concludes the paper.
\section{Real-time Control of a BESS}
\begin{figure*}[!htbp]
	\vspace{-0.4cm}
	\centering
	\includegraphics[width=\linewidth]{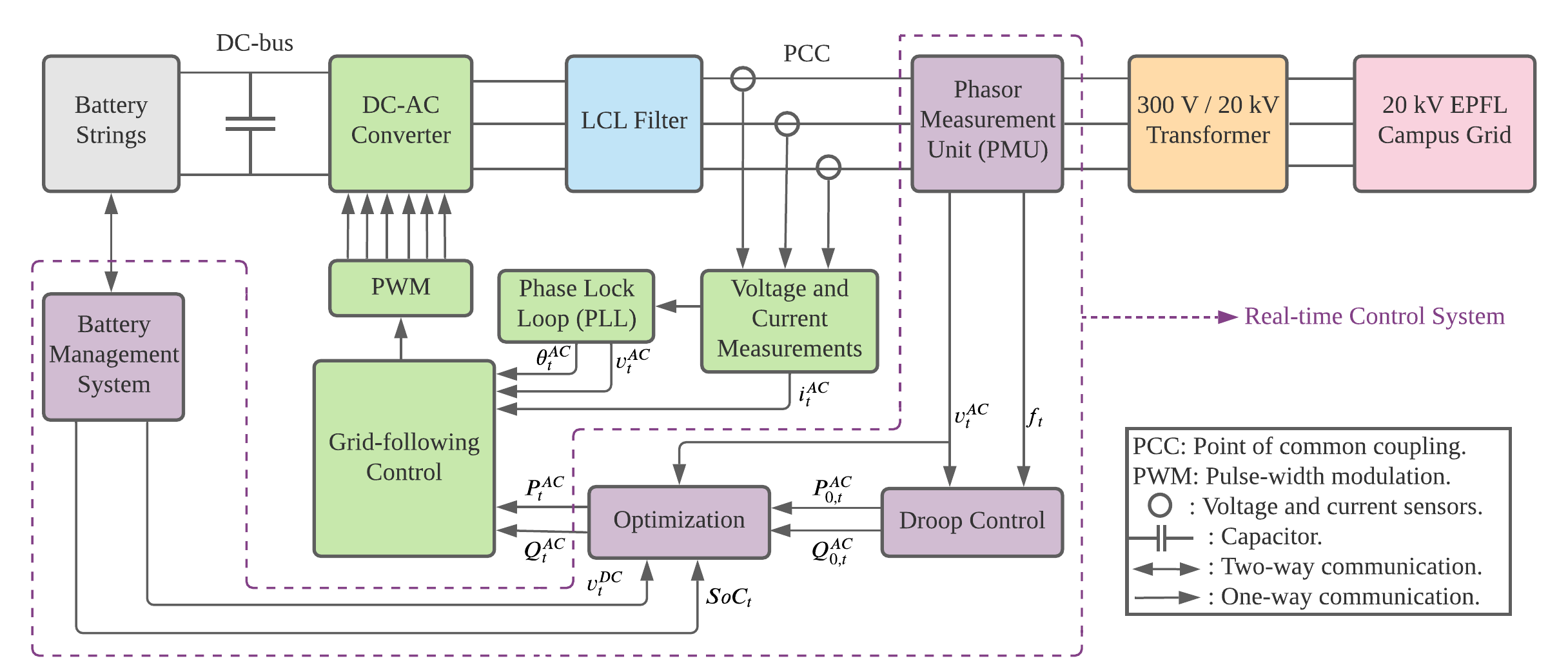}
	\caption{Block diagram of the BESS control system.}
	\label{fig:control_block}
\end{figure*}
As known, there are two main approaches to control BESS power converters: the \textit{grid-following} and the \textit{grid-forming} controls \cite{pfbc_micisl,conpowcon_acmic}. A grid-following unit controls the converter in order to inject active and reactive powers by modifying the amplitude and angle (with respect to the grid voltage phasor) of the converter reference current. To do so, it requires the knowledge of the fundamental phasor of the grid voltage at a point of common coupling (PCC). On the contrary, a grid-forming unit controls the frequency and voltage at the PCC in order to imitate a voltage source behind an impedance. To do so, this control does not require the knowledge of the fundamental frequency phasor of the grid voltage at the PCC.
A grid-following unit can also provide \textit{frequency and voltage supports} by adding further high-level control loops to generate the reference active and reactive powers  \cite{conpowcon_acmic}. In this paper we consider this former control since, in actual power system, the majority of converter-interfaced BESSs are controlled as grid-following ones. We show in Fig. \ref{fig:control_block} the structure of the BESS control system. The proposed real-time control system works as the high-level control layer to provide the optimal power set-point for the grid-following control of the DC-AC converter. The real-time control system is denoted in purple color in Fig. \ref{fig:control_block}. It consists of the battery management system, phasor measurement unit (PMU), droop control and optimization. The battery management system is responsible for monitoring and controlling of the battery cells. The measured DC-bus voltage and battery state of charge (SoC) are communicated to the optimization module which solves the optimization problem we describe in the following to provide the optimal power set-point to the grid-following control of the DC-AC converter. The initial power set-point is calculated by the droop control. The PMU is used to accurately estimate the grid voltage and frequency.
\subsection{Optimization Model}
Our DC-AC converter is working in a grid-following control with the supporting mode where we can specify the power set-point of the converter. The power set-points are given by an outer control loop i.e. the real-time controller that we proposed in this paper.
Traditionally the frequency control and voltage support to the power grid are provided by droop control formulated as \eqref{eq:droop_p}-\eqref{eq:droop_q} \cite{kundur1994power}.
\begin{subequations}
\begin{align}
    P^{AC}_{0,t}&= \alpha \Delta f_t = \alpha (f_t-f^{nom}),\, \forall \left | \Delta f_t\right |>\Delta^{min} f_t\label{eq:droop_p} \\
    Q^{AC}_{0,t}&= \beta \Delta v^{AC}_t = \beta (v^{AC}_t-v^{ACnom}),\, \forall \left | \Delta v^{AC}_t\right |>\Delta^{min} v^{AC}_t\label{eq:droop_q}
\end{align}
Where $(P^{AC}_{0,t}, Q^{AC}_{0,t})$ are the initial active and reactive power set-points of the BESS at time-step $t\in T$ being the operation period. 
$\Delta f_t=f_t-f^{nom}$ is the power grid frequency deviation from the nominal frequency $f^{nom}$. 
$\Delta v^{AC}_t=v^{AC}_t-v^{ACnom}$ is the power grid direct sequence AC voltage magnitude deviation from the nominal voltage $v^{ACnom}$. The parameters $\alpha_0,\beta_0$ are the initial droop coefficients. 
$\Delta^{min} f_t, \Delta^{min}v^{AC}_t$ are the dead-bands of frequency deviation and voltage deviation {of} the droop control. 
Note {that} the dead-bands are necessary mainly for two reasons: (1) they are required by grid codes (e.g. Germany \cite{entsoefsm}), (2) to avoid {unnecessary control actions}  of the controller.
The droop coefficients are set by $\alpha =\frac{P^{max}}{\Delta^{max} f_t}, \beta= \frac{Q^{max}}{\Delta^{max} v^{AC}_t}$ to maximize the frequency control and voltage support performance.
The maximum active power and reactive power parameters $P^{max}, Q^{max}$ are specified by the BESS. Historical measurements can be used to find the maximum frequency deviation $\Delta^{max} f$ and maximum voltage deviation $\Delta^{max} v^{AC}$ parameters. These parameters can be found, for instance, in \cite{zecchino2019optimal}. 
\begin{figure}[!htbp]
	\vspace{-0.5cm}
		\centering
		\hspace*{-0.45cm}
    	\includegraphics[width=0.8\linewidth, angle =90]{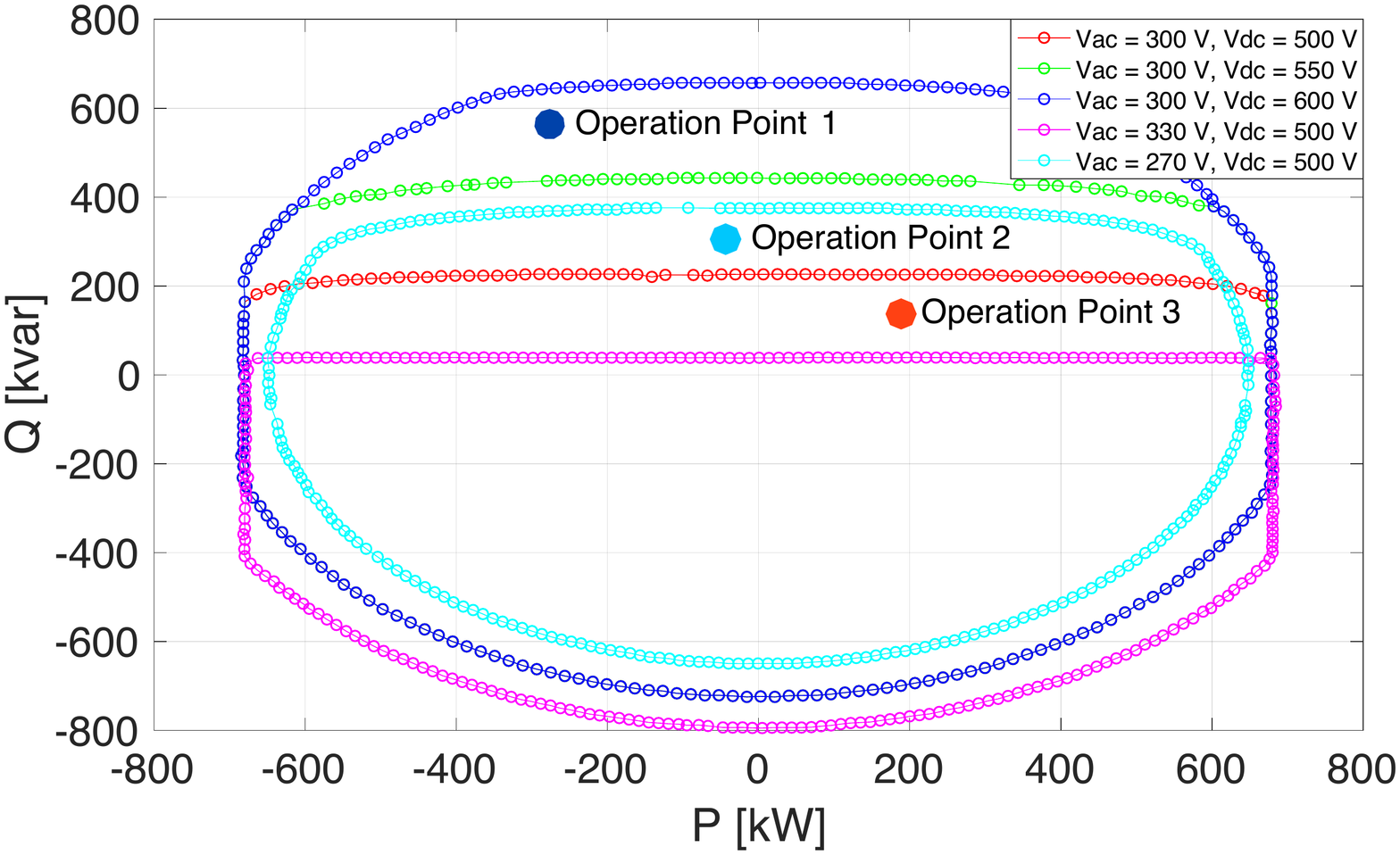}
    	\vspace{-1.4cm}
		\caption{An example of voltage-dependent converter capability curves as a function of AC-side and DC-side voltages \cite{zecchino2019optimal}.}
		\label{fig:Converter_CapabilityCurves}
\end{figure}

We consider the converter capability curve shown in Fig. \ref{fig:Converter_CapabilityCurves}, together with the battery cells' security requirements as a series of constraints in our real-time control as in \eqref{eq:bound_pq}.  The x-axis in Fig. \ref{fig:Converter_CapabilityCurves} is the active power (P). The y-axis in Fig. \ref{fig:Converter_CapabilityCurves} is the reactive power (Q). The capability curve is determined jointly by the AC-side voltage and DC-side voltage (see the legend) of the DC-AC converter. We use different colors to denote the difference capability curves for difference DC-side and AC-side voltages. The feasible operation points of the DC-AC converter must be always inside the capability curve. We put three operation points inside the capability curves as examples in Fig. \ref{fig:Converter_CapabilityCurves}. The operation points with different colors in Fig. \ref{fig:Converter_CapabilityCurves} are feasible with respect to the corresponding capability curves to avoid the converter misoperation. The capability curve of DC-AC converters is {usually} derived by considering the apparent power capability, the DC-bus voltage limit and the AC-side voltage limit. The capability curve is usually given by the converter's manufacturer and provided to the BESS user as a technical specification. For the expression and parameter of \eqref{eq:bound_pq}, please refer to reference \cite{zecchino2019optimal}.
\begin{align}
    h(P^{AC}_t, Q^{AC}_t, v^{DC}_t, v^{AC}_t, SoC_t) \leq 0 \label{eq:bound_pq}
\end{align}
Where $(P^{AC}_t, Q^{AC}_t)$ is the active and reactive power set-point of the BESS, $v^{DC}_t, v^{AC}_t$ are the voltage magnitudes {on} the DC bus and AC side. These parameters are measured by dedicated voltage sensors available in the BESS asset. $SoC_t$ is the state-of-charge of the battery cells estimated by the battery management system. Instead of using the approximated capability curve expressed as $(P^{AC}_t)^2+(Q^{AC}_t)^2 \leq (S^{ACmax})^2$ (where $S^{ACmax}$ is the maximum apparent power of the converter), we fit the actual converter capability curves $h$ by using measurements from the converter {management system} and then scale the capability curves proportionally according to the actual useful capacity of the BESS. The fitted converter capability curves are represented by a series of linear and quadratic equations. We use (1c) instead of the approximated capability curve because (I) (1c) is more accurate to capture the feasible operating points of the BESS converter; (II) the capability curve is actually varying according to different converter DC-side voltage and AC-side voltage. In other words, when the DC-side voltage or AC-side voltage changes, the corresponding capability curve will also change. Our BESS controller is capable of selecting the correct capability curve in real-time operations. We show in Fig. \ref{fig:compare} the differences between the approximated circle-shape converter capability curve (in blue color) and the fitted actual capability curve. As it can be seen, there are several regions (arrowed in Fig. \ref{fig:compare}) covered by the approximated capability curve that are actually infeasible for the actual capability curve.

The battery state of charge $SoC_{t}$ is updated according to equation \eqref{eq:socupd}.
\begin{align}
SoC_{t+1}&=SoC_{t}+ \frac{\int_{t}^{t+1}i^{DC}_{t} dt}{C^{max}}\approx SoC_{t}+ \frac{P^{DC}_{t}}{v^{DC}_{t} C^{max}} \Delta t \label{eq:socupd} 
\end{align}
Where $C^{max}$ is the maximum storage capacity of the battery in Ampere-hour i.e. [A·h]. $C^{max}$ can be updated at the beginning of each control loop to estimate the actual storage capacity according to the method {proposed} in \cite{mcont_micogrid}. The data-set of $C^{max}$ at different charging/discharging {rates} are usually available from the battery manufacturer \cite{LecCell}. A look-up table of $C^{max}$ at any possible level of charging/discharging current is obtained by linearly interpolating this data-set {with respect to} the current value of the previous control loop. $i^{DC}_t\approx \frac{P^{DC}_t}{v^{DC}_{t}}$ is the charging/discharging DC current. $\Delta t = 50\text{ ms}$ is the time resolution of each control loop. Based on the initial battery state of charge measurement $SoC_{0}$, $SoC_{t}$ can be calculated during the BESS operation. The battery output powers are jointly constrained by the DC-AC converter capability curves and the state-of-charge constraint. These constraints avoid the over charging/discharging of the battery.
\begin{figure}[!htbp]
	\vspace{-0.4cm}
	\centering
	\includegraphics[width=\linewidth]{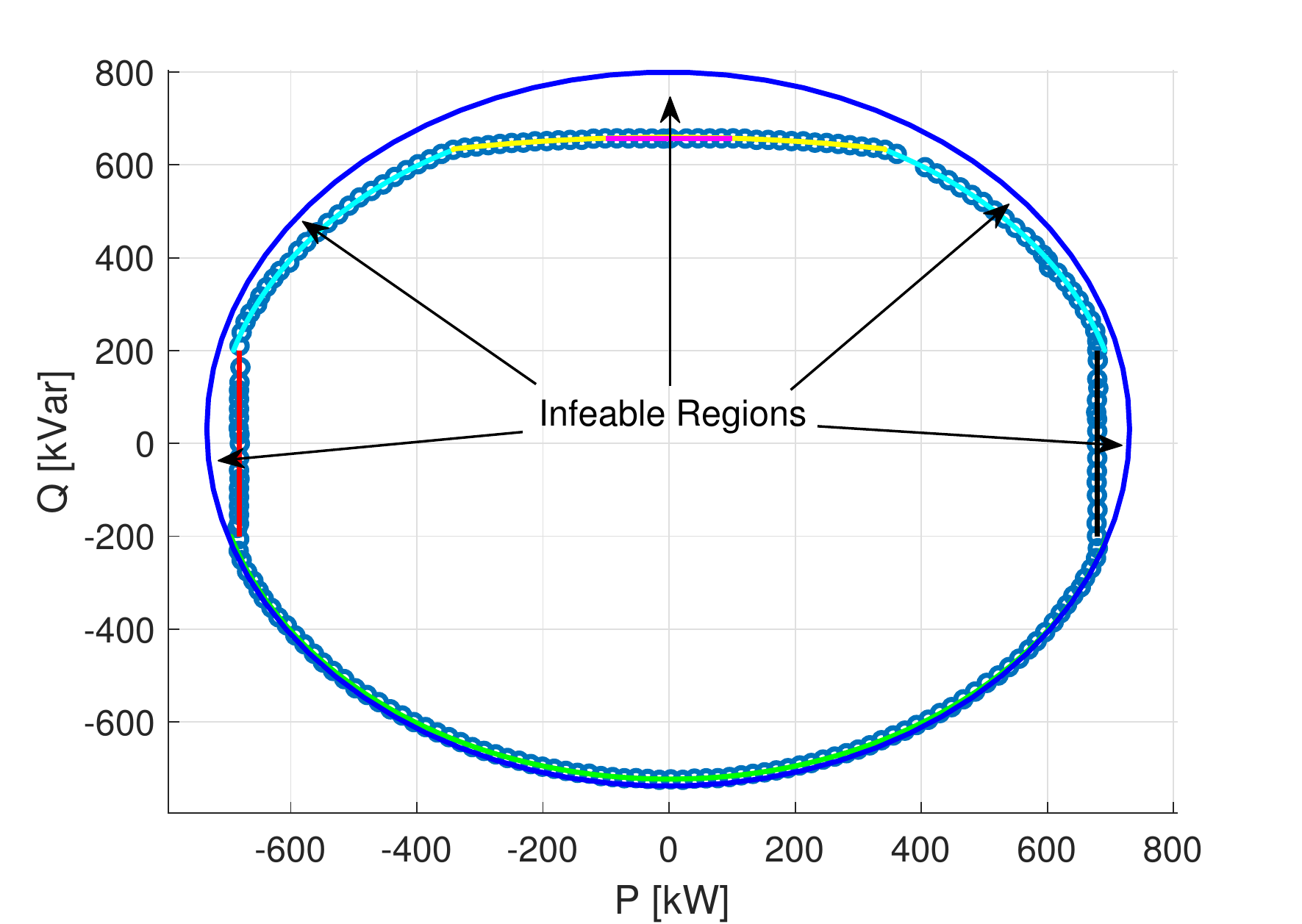}
	\vspace{-0.6cm}
	\caption{{Comparison of the approximated and actual converter capability curves.}}
	\label{fig:compare}
\end{figure}
The relationship of the active power at the DC bus $P^{DC}_t$ and the active power at the AC side of the converter $P^{AC}_t$ is expressed as:
\begin{align}\label{eq:pdcac}
P^{DC}_t=\left\{\begin{matrix}
\eta P^{AC}_t ,\; \forall P^{AC}_t<0\\ 
\frac{P^{AC}_t}{\eta},\; \forall P^{AC}_t\geq0 
\end{matrix}\right. 
\end{align}
Where $\eta$ is the charging/discharging efficiency of {the} converter. $P^{AC}_t<0$ means charging of the BESS and $P^{AC}_t\geq 0$ means discharging. 
We estimate the battery {state} by using a three time constant model (TTC) shown in Fig. \ref{fig:ttc_model}.
\begin{figure}
	\vspace{-0.2cm}
    \centering
    \includegraphics[width=\linewidth]{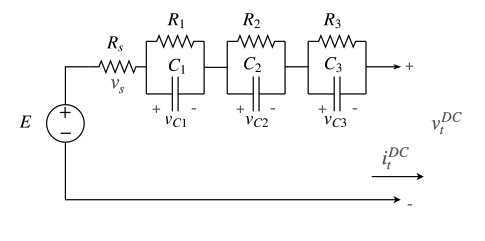}
    \vspace{-1.1cm}
    \caption{Battery TTC model \cite{zecchino2019optimal}. }
    \label{fig:ttc_model}
\end{figure}
\begin{align}
    &C_{1(2,3)}\frac{\mathrm{d}v_{C1(2,3)} }{\mathrm{d} t}+\frac{v_{C1(2,3)}}{R_{1(2,3)}}=\frac{v_s}{R_s} \label{eq:ttcvc1}\\
    &v_s+v_{C1}+v_{C2}+v_{C3}=E-v^{DC}_{t} \label{eq:ttcsum}
\end{align}
Where $\bold{v_c}=[v_{C1};\, v_{C2};\, v_{C3}]$ are the state variables i.e. voltages that are updated by solving \eqref{eq:ttcvc1}-\eqref{eq:ttcsum} in each control loop. To accelerate the computational efficiency, we use a discrete model to update $\bold{v_c}$.
The TTC model capacitance parameters $C_1, C_2, C_3$ and resistance parameters $R_s, R_1, R_2, R_3$ are identified by generating random length time duration power set-points and, then by measuring the corresponding battery voltage and current. 
Details about this process and the identified parameters can be found in \cite{zecchino2019optimal}. 
The voltage variable $E$ represents the open circuit voltage of the battery, which depends on the SoC as 
$E=a + b\times SoC_t$. The parameters $a$ and $b$ are identified within the TTC model identification process or provided by the BESS manufacturer. An example of the identified parameters can be found in \cite{zecchino2019optimal}. 

After updating $\bold{v_c}=[v_{C1};\, v_{C2};\, v_{C3}]$, {substituting $v_s=\frac{P^{DC}_t}{v^{DC}_t} R_s$ in equation \eqref{eq:ttcsum} and multiplying both sides {by} $v^{DC}_t$}, equation \eqref{eq:ttcsum} is {rewritten} {{as}}:
\begin{align}
    (v^{DC}_t)^2+(\bold{1}^T\bold{v_c}-E)v^{DC}_t+P^{DC}_t R_s=0 \label{eq:vdcpdc}
\end{align}
Where $\bold{1}^T=[1, 1, 1]$. {$\bold{1}^T\bold{v_c}=v_{C1}+v_{C2}+v_{C3}$}. Solving constraints \eqref{eq:bound_pq} jointly with equation \eqref{eq:vdcpdc} gives the feasible power set-points satisfying the evolving capability curves during the control loop. The operational requirement of the DC-bus voltage $v^{DC}_{t}$ for the converter is:
\begin{align}
    v^{DCmin} \leq v^{DC}_t \leq v^{DCmax} \label{eq:bound_vdc}
\end{align}

The security range $(SoC^{min}, SoC^{max})$ of $SoC_t$ should be always kept during all the operational periods:
\begin{align}
    SoC^{min}\leq SoC_t \leq SoC^{max} \label{eq:bound_soe}
\end{align}

The optimal active and reactive power set-points are obtained by solving the following optimization problem: 
\begin{align}
    &\underset{\Omega}{\text{Minimize}}\; obj=( P^{AC}_{t}- P^{AC}_{0,t})^2+( Q^{AC}_{t}- Q^{AC}_{0,t})^2 \label{eq:opt_pq}\\
    &\text{subject to}\;\; \eqref{eq:bound_pq}-\eqref{eq:pdcac},\eqref{eq:vdcpdc}-\eqref{eq:bound_soe}\nonumber
\end{align}
This optimization model is necessary because the initial power set-points $(P^{AC}_{0,t}, Q^{AC}_{0,t})$ may be outside the feasible operational region of the BESS. We show in the experimental results of this paper how the optimization model works when this happens. We denote this original optimization problem as PQ-opt-o. Where $\Omega=\left \{P^{AC}_{t}, Q^{AC}_{t}, v^{DC}_{t}, P^{DC}_{t}\right \}$ is the set of variables. $\Phi=\left \{P^{AC}_{0,t}, Q^{AC}_{0,t}, v^{AC}_{t}, \eta, SoC_{t}\right \}$ is the set of parameters from calculations or measurements. So, the optimal power set-point $(P^{AC}_{t}, Q^{AC}_{t})$ is the one set-point most close to the initial power set-points $(P^{AC}_{0,t}, Q^{AC}_{0,t})$ inside the feasible operational region of the BESS defined by \eqref{eq:droop_p}-\eqref{eq:bound_soe}. If the initial power set-point $(P^{AC}_{0,t}, Q^{AC}_{0,t})$ is already inside the feasible operational region of the BESS, the optimal power set-point $(P^{AC}_{t}, Q^{AC}_{t})$ is equal to $(P^{AC}_{0,t}, Q^{AC}_{0,t})$. Otherwise, if the initial power set-point $(P^{AC}_{0,t}, Q^{AC}_{0,t})$ is not feasible for the BESS, the optimal power set-point $(P^{AC}_{t}, Q^{AC}_{t})$ is computed by (1k). 
\begin{figure}[h]
		\centering
		\vspace{-0.4cm}
    	\includegraphics[width=0.8\linewidth]{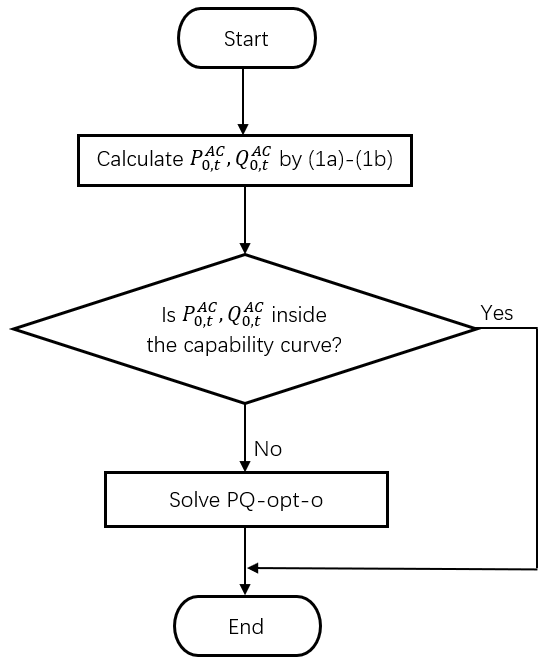}
    	\vspace{-0.3cm}
		\caption{Flow chart of the BESS real-time control.}
		\label{fig:flowchart}
\end{figure}
\end{subequations}
\subsection{Reformulation and Solution Algorithm {for the PQ-opt-o Problem}}
The optimization problem PQ-opt-o is nonconvex due to the nonconvex constraints \eqref{eq:pdcac} and \eqref{eq:vdcpdc}. To improve the solution quality, we can firstly replace constraint \eqref{eq:pdcac} by \eqref{eq:pdcac0}.
\begin{subequations}
\begin{align}\label{eq:pdcac0}
P^{DC}_t=\left\{\begin{matrix}
\eta P^{AC}_t ,\; \forall P^{AC}_{0,t}<0\\ 
\frac{P^{AC}_t}{\eta},\; \forall P^{AC}_{0,t}\geq0 
\end{matrix}\right. 
\end{align}
Constraint \eqref{eq:pdcac0} is equivalent {to} \eqref{eq:pdcac} since $P^{AC}_{0,t}$ always has the same sign with $P^{AC}_{t}$. So, instead of using constraint \eqref{eq:pdcac}, we can use either $P^{DC}_t=\eta P^{AC}_t$ or $P^{DC}_t=\frac{P^{AC}_t}{\eta}$ to formulate the the optimization model based on $P^{AC}_{0,t}$. 

To efficiently find an optimal solution, we then convexify constraint \eqref{eq:vdcpdc} to \cite{SOCACOPF_relax_fea_me,DLMP_me}:
\begin{align}
    (v^{DC}_t)^2+(\bold{1}^T\bold{v_c}-E)v^{DC}_t+P^{DC}_t R_s\leq0 \label{eq:vdcpdc2}
\end{align}
However, the convex relaxation in \eqref{eq:vdcpdc2} can make the final solution infeasible {with respect to} the original constraint \eqref{eq:vdcpdc}. In order to {force} the optimal solution to be feasible, we propose to modify the original objective function to:
\begin{align}
    obj^{M}=( P^{AC}_{t}- P^{AC}_{0,t})^2+( Q^{AC}_{t}- Q^{AC}_{0,t})^2 - \xi v^{DC}_t \label{eq:opt_pq2}
\end{align}
We denote the modified optimization problem as PQ-opt-m:
\begin{align}
    \underset{\Omega}{\operatorname{argmin}}\;obj^M:= \left\{\Omega \in [\eqref{eq:bound_pq}-\eqref{eq:socupd},\eqref{eq:bound_vdc}-\eqref{eq:bound_soe},\eqref{eq:pdcac0}-\eqref{eq:vdcpdc2}]\right\}
\end{align}

{It is worth noting that,} PQ-opt-m is convex. {Furthermore,} we propose the following theorem 1 to explain the equivalence between the modified optimization problem PQ-opt-m and the original optimization problem PQ-opt-o.

\begin{theorem}\label{theom:bess_relax1}
If the original optimization problem PQ-opt-o is feasible and $\xi>0$, the modified optimization problem PQ-opt-m is also feasible.
\end{theorem}
\begin{proof}
We prove theorem \ref{theom:bess_relax1} by showing that any feasible solution of the original optimization problem PQ-opt-o is always feasible for the modified optimization problem PQ-opt-m. In the meanwhile, we prove there is a lower bound for the objective function of PQ-opt-m. 

Suppose one feasible solution of PQ-opt-o is $\Omega^{0}=\left \{P^{AC0}_{t}, Q^{AC0}_{t}, v^{DC0}_{t}, P^{DC0}_{t}\right \}$. Since $\Omega^{0}$ satisfies all the constraints of PQ-opt-o and constraint \eqref{eq:vdcpdc} is a subset of constraint \eqref{eq:vdcpdc2}, $\Omega^{0}$ is also feasible for PQ-opt-m. 

For the lower bound $f^{Mmin}$  of the objective function of PQ-opt-m, since $\xi>0, v^{DC}_t<v^{DCmax},( P^{AC}_{t}- P^{AC}_{0,t})^2+( Q^{AC}_{t}- Q^{AC}_{0,t})^2\geq0$, we have:
\begin{align}
    obj^{M}\geq 0- \xi v^{DCmax}_t=- \xi v^{DCmax}_t=f^{Mmin}
\end{align}
\end{proof}
\begin{theorem}\label{theom:bess_relax2}
If the original optimization problem PQ-opt-o is feasible, $\xi>0$ and $\frac{\partial g}{\partial v^{DC_t}}>0$, the optimal solution of the modified optimization problem PQ-opt-m is equal to the global optimal solution of the original optimization problem PQ-opt-o.
\end{theorem}
\begin{proof}
Since the original optimization problem PQ-opt-o is supposed to be feasible, the modified optimization problem is also feasible according to theorem \ref{theom:bess_relax1}.
We denote the optimal solution of the modified optimization problem PQ-opt-m as $\Omega^{*}=\left \{P^{AC*}_{t}, Q^{AC*}_{t}, v^{DC*}_{t}, P^{DC*}_{t}\right \}$. We firstly prove that $\Omega^{*}$ is also feasible for PQ-opt-o. Constraint \eqref{eq:vdcpdc} is the only constraint we need to consider since $\Omega^{*}$ by definition already satisfies all the other constraints of PQ-opt-o. Suppose $\Omega^{*}$ is not feasible for constraint \eqref{eq:vdcpdc}, we have:
\begin{align}
    g=(v^{DC*}_t)^2+(\bold{1}^T\bold{v_c}-E)v^{DC*}_t+P^{DC*}_t R_s<0 \label{eq:vdcpdc3}
\end{align}
We derive the first-order derivative of the left side of constraint \label{eq:vdcpdc3} as:
\begin{align}
    \frac{\partial g}{\partial v^{DC_t}}=2v^{DC*}_t+(\bold{1}^T\bold{v_c}-E)\label{eq:gvdcE}
\end{align}
According to the battery TTC model in Fig. \ref{fig:ttc_model}, we have:
\begin{align}
    \bold{1}^T\bold{v_c}-E=-v_s-v^{DC}_t \label{eq:Evc}
\end{align}
Substituting \eqref{eq:Evc} to \eqref{eq:gvdcE} gives:
\begin{align}
    \frac{\partial g}{\partial v^{DC_t}}=v^{DC*}_t-v_s\label{eq:gvdcvs}
\end{align}
Since $\frac{\partial g}{\partial v^{DC_t}}>0$, there exists $v^{DC**}_t>v^{DC*}_t$ such that:
\begin{align}
    g=(v^{DC**}_t)^2+(\bold{1}^T\bold{v_c}-E)v^{DC**}_t+P^{DC*}_t R_s=0 \label{eq:vdcpdc4}
\end{align}
Because constraint \eqref{eq:vdcpdc4} is a subset of constraint \eqref{eq:vdcpdc2}, $v^{DC**}_t$ also satisfies constraint \eqref{eq:vdcpdc2}. Since $v^{DC**}_t>v^{DC*}_t$ and $\xi>0$, $v^{DC**}_t$ can give a smaller objective function value $f^{m**}$. This means {$\Omega^{*}$} is not the optimal solution which contradicts our initial assumption. So, $\Omega^{*}$ must be feasible for PQ-opt-o.

We then prove theorem \ref{theom:bess_relax2} by showing that $\Omega^{*}$ is also optimal for PQ-opt-o. Suppose the global optimal solution of PQ-opt-o is $\Omega^{'}=\left \{P^{AC'}_{t}, Q^{AC'}_{t}, v^{DC'}_{t}\right \}$ and $\Omega^{'}\neq\Omega^{*}$. This means the objective function value $f^{'}=f(\Omega^{'})<f^{*}=f(\Omega^{*})$. According to theorem \ref{theom:bess_relax1}, $\Omega^{'}$ is also feasible for PQ-opt-m. We can construct another feasible solution $\Omega^{''}$ of PQ-opt-m by $P^{AC''}_{t}:=P^{AC'}_{t},\,Q^{AC''}_{t}:=Q^{AC'}_{t},\,v^{DC''}_{t}:=v^{DC*}_{t}$.
If $P^{AC''}_{t}\neq P^{AC*}_{t}$ or $Q^{AC'}_{t}\neq Q^{AC*}_{t}$, $\Omega^{''} \neq \Omega^{*}$ and $f^{m''}=f^m(\Omega^{''})<f^{m*}=f^m(\Omega^{*})$. This contradicts our assumption that $\Omega^{*}$ is the optimal solution of PQ-opt-m. So $P^{AC'}_{t}= P^{AC*}_{t}$ and $Q^{AC'}_{t}= Q^{AC*}_{t}$ must hold. This means, according to constraints \eqref{eq:pdcac}-\eqref{eq:vdcpdc}, $v^{DC'}_t=v^{DC*}_t$ is valid. In other words $\Omega^{'}=\Omega{*}$ which contradicts our assumption that they are not equal. So the initial assumption cannot hold, the optimal solution of PQ-opt-m is also global optimal for PQ-opt-o.
\end{proof}
In this paper, the per unit values are according to the base power {of} 720 kVA and base voltage {of} 700 V. Theorem \ref{theom:bess_relax2} is valid for our BESS because the maximum charge/discharge current is $i^{DCmax}=\frac{S^{ACmax}}{v^{DCmin}_t}=\frac{720 \,\text{kVA}}{600\, \text{V}}= 1.167\, p.u.$
Note this is an estimated value that is actually much larger than the maximum allowed value of charge/discharge current {of the adopted BESS}. According to the identified parameters of the battery TTC model, $R_s<0.045\, p.u.$, this means the maximum voltage $v_s$ in the battery TTC model is less than $v^{max}_s=i_{DCmax} R_s=0.051 \,p.u.$. Considering $V^{DCmin}=0.86\, p.u.$, $\frac{\partial g}{\partial v^{DC_t}}=v^{DC*}_t-v_s>0$. 
According to theorems \ref{theom:bess_relax1}-\ref{theom:bess_relax2}, we can find the global optimal solution of the original nonconvex optimization problem PQ-opt-o by solving the modified convex optimization problem PQ-opt-m. This gives us better computational efficiency {for} the real-time control of {the} BESS.
In summary, we propose the following solution algorithm \ref{alg:heuris}.
\begin{algorithm}\label{alg:heuris}
 \caption{Optimization Solution Algorithm}
\SetAlgoLined
\KwResult{Optimal Power Set-Points $P^{AC}_{t}, Q^{AC}_{t}$}
 Initialization $P^{AC}_{0,t}, Q^{AC}_{0,t}$ Based on Equations \eqref{eq:droop_p}-\eqref{eq:droop_q}\;
  \lIf{$P^{AC}_0<0$}{
   Replace \eqref{eq:pdcac} by $P^{DC}_t=\eta P^{AC}_t$
   }\lElse{Replace \eqref{eq:pdcac} by $P^{DC}_t=\frac{P^{AC}_t}{\eta}$\bf{end}}
   {Use fmincon Solver in MATLAB to} Solve the Modified Optimization Problem PQ-opt-m\\
\end{algorithm}

{The first stage of this algorithm selects the correct charging/discharging equations based on the sign of the initial power set-points. In this way, we can avoid the usage of binary variables (which makes it harder to solve optimization problems) in formulating the optimization models. No iterations are required in this stage. Then we make use of the \textit{fmincon} solver in MATLAB after the initial processing of the optimization model. Since PQ-opt-m is convex, the \textit{fmincon} solver is sufficient to find the global optimal solution of the problem.}

\end{subequations}
 \subsection{Discretization}
The solution {Algorithm} \ref{alg:heuris} solves PQ-opt-m and performs better than solving the PQ-opt-o. We propose to further improve the computational efficiency of this algorithm by discretizing the feasible region of PQ-opt-m. The major motivation is to avoid solving the optimization problem online in the real-time control of a BESS. This consideration is practical because: (1) nonlinear optimization solvers may not be available or too expensive in many real-world applications, (2) real-time control requires stable computational performance. 
In order to provide an adequate frequency control, we should achieve a real-time control latency down to 100 ms, we can approximate the DC-bus voltage by using the real-time measurement instead of using the battery TTC model. In other words, $v^{DC}_t$ is now a parameter in PQ-opt-m as $v^{AC}_t$. Similarly, for the battery SoC constraints \eqref{eq:socupd}-\eqref{eq:bound_soe}, we can directly use the battery SoC measurement to approximate the battery SoC during the control loop. In this way, we can approximate the voltage-dependent optimization problem PQ-opt-m with the static one. We propose to solve the following static optimization problem in order to discretize the feasible region of PQ-opt-m.
\begin{subequations}
\begin{align}
    \underset{\Omega^{S}}{\operatorname{argmax}}\;obj^{S}:= \left\{\Omega^{S} \in [\eqref{eq:bound_pq},\eqref{eq:bound_soe},\eqref{eq:pdcac0}]\right\}
\end{align}
Where $\Omega^{S}=\left\{P^{AC}_t,Q^{AC}_t\right\}$ is the set of optimization variables. The objective function is formulated as $f^{S}=(P^{AC}_t)^2+(Q^{AC}_t)^2$. We denote this optimization problem as PQ-opt-s. Sequentially solving this optimization problem offline can quantify the feasible region of PQ-opt-m by using the length of the vector $\left | S \right |^{max}=\left | (P^{AC}_t, Q^{AC}_t) \right |^{max}$ shown in Fig. \ref{fig:vector_compare2}. 
In this way, we discretize the feasible region of PQ-opt-m by a $1^{\circ}$ resolution over the $360^{\circ}$ two dimension active power and reactive power $(P^{AC}_t, Q^{AC}_t)$ plane. In other words, we solve PQ-opt-s 360 times at $\frac{Q^{AC}_t}{P^{AC}_t}=tan(1^{\circ}),tan(2^{\circ}),...tan(360^{\circ})$. Where $tan(\cdot)$ is the tangent function. We then get the values of $\left | {S}_\theta \right |^{max}$ as an array for $\theta={1^{\circ},2^{\circ},...360^{\circ}}$. Note the discretization resolution can be smaller according to the accuracy requirement {{of the modeler}}. {{Obviously}}, {smaller resolution of the discretization leads to higher accuracy.} The proposed discretization methodology is applicable for any given discretization resolution. After obtaining $\left | S_\theta \right |^{max}$ offline, we propose algorithm \ref{alg:fast} for the real-time control of BESS:
\begin{figure}
\vspace{-0.5cm}
\hspace*{-0.4cm}
    \centering
    \includegraphics[width=0.8\linewidth, angle =90]{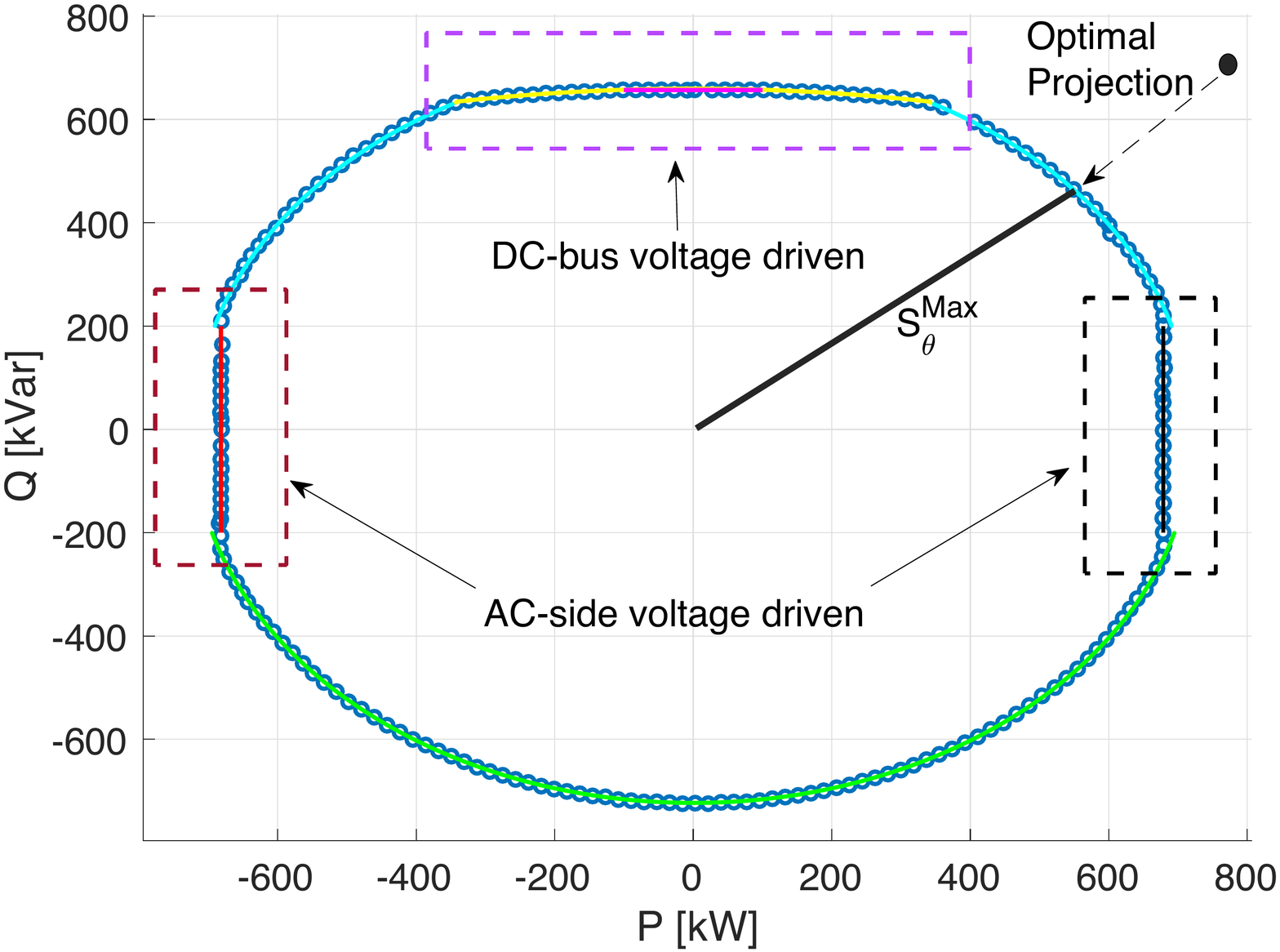}
    \vspace{-0.6cm}
    \caption{{Optimal power set-point projection by using PQ-opt-o(m) and discretization of the BESS feasible region in the PQ-plane.}}
    \label{fig:vector_compare2}
\end{figure}
\begin{algorithm}\label{alg:fast}
 \caption{Fast Real-time Control Algorithm}
\SetAlgoLined
\KwResult{Optimal Power Set-Points $P^{AC}_{t}, Q^{AC}_{t}$}
 Initialization $P^{AC}_{0,t}, Q^{AC}_{0,t}$ Based on Equations \eqref{eq:droop_p}-\eqref{eq:droop_q}\;
 \eIf{$P^{AC}_{0,t}=0$}{
     \lIf{$Q^{AC}_{0,t}\geq0$}{$\theta=90^{\circ}$
     }\lElse{$\theta=270^{\circ}\;$\bf{end}}
 }{
 \lIf{$P^{AC}_{0,t}>0$ and $Q^{AC}_{0,t}\geq0$}{$\theta=\left \lceil{arctan(\frac{Q^{AC}_{0,t}}{P^{AC}_{0,t}})}\right \rceil$ 
 }\lElse{
 $\theta=\left \lceil{arctan(\frac{Q^{AC}_{0,t}}{P^{AC}_{0,t}})}\right \rceil+180^{\circ}\;$\bf{end}}
 }
 \lIf{$\left | (P^{AC}_{0,t}, Q^{AC}_{0,t}) \right |\leq \left | {S}_\theta \right |^{max}$}{$(P^{AC}_{t}, Q^{AC}_{t})=(P^{AC}_{0,t}, Q^{AC}_{0,t})$}
 \lElse{$P^{AC}_{t}={S}_\theta cos(\theta)$,$\;Q^{AC}_{t}={S}_\theta sin(\theta)$ \bf{end}}
\end{algorithm}
Where $\left \lceil \cdot \right \rceil$ is the ceil function which returns the minimum integer larger than $\cdot $, $cos(\cdot),sin(\cdot)$ are the cosine and sine functions. We show later in the experiment section the comparative performance of algorithms \ref{alg:heuris} and \ref{alg:fast}.
{In Algorithm 2, there are two stages. The first stage is to select the correct angle for the discretized region of the capability curve. The second stage is to give the acceptable power set-point by comparing the length of the complex power vector with maximum allowed length of the complex power vector. Both stages do not require iterations.}
\end{subequations}
\section{Experimental Results}
\begin{figure}[!htbp]
	\vspace{-0.1cm}
    \centering
    \includegraphics[width=\linewidth]{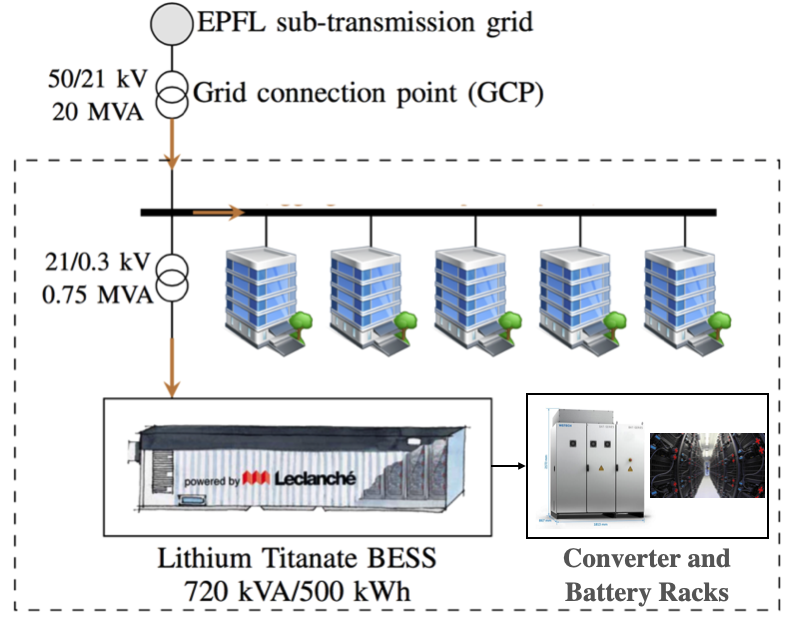}
    \caption{BESS installation set-up on the EPFL campus.}
    \label{fig:bess_epfl_ins}
\end{figure}
\begin{figure}[!htbp]
	\hspace*{-0.2cm}
    \centering
    \includegraphics[width=\linewidth]{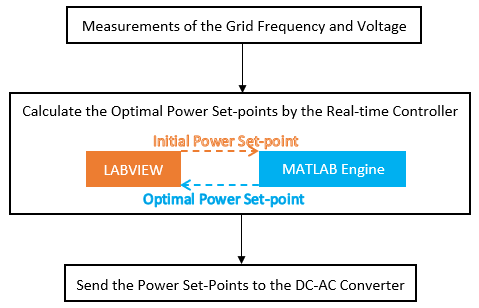}
    \vspace{-0.7cm}
    \caption{Measurement, software interface and communication of the real-time control.}
    \label{fig:inter}
\end{figure}
In order to validate the proposed algorithms, we have carried out experiments on the 720 kVA / 560 kWh BESS installed on EPFL campus in Lausanne, Switzerland. The BESS installation set-up is shown in Fig. \ref{fig:bess_epfl_ins}. The battery technology is Lithium-Titanate-Oxide (LTO). A 4-quadrant DC-AC converter is used to connect the DC-bus of the battery through a 300 V / 20 kV {step-up} transformer to the EPFL campus power grid. The parameters of the BESS and battery TTC model can be found in \cite{zecchino2019optimal}. The real-time control and monitor system of the BESS is developed in Labview 2018. We use Modbus TCP/IP protocol for the communication. The DC-AC converter status, battery cell status and AC side grid information are monitored in the real-time control loop. Algorithm \ref{alg:heuris} is coded within the YALMIP optimization tool \cite{Lofberg2004}. Both algorithms \ref{alg:heuris}-\ref{alg:fast} are executed in the MATLAB runtime engine of Labview. Fig. \ref{fig:inter} shows the interface between the measurement, Labview, MATLAB and the DC-AC converter. Algorithm \ref{alg:heuris} {requires} 200 ms to update the real-time control loop, while algorithm \ref{alg:fast} {requires} 100 ms. We show the performance of the proposed real-time control {to provide} frequency control and voltage support services to the grid. Note that, when the initial power set-points are outside the feasible region of the DC-AC converter capability curve, the protection set-up by the manufacturer of the BESS automatically reset the power set-points to zero if Algrorithm 1 or 2 is not used.
We run a one-hour experiment for each algorithm \ref{alg:heuris} and algorithm \ref{alg:fast}. In our experiment, we are using a BATI720 converter which requires $v^{DCmin}=600\, \text{V}, v^{DCmax}=800\, \text{V}$. We set $SoC^{min}=10\%, SoC^{max}=90\%$. 

To show the effectiveness of the proposed algorithms, we conduct experiments for two scenarios. In the first scenario, the droop coefficients are $\alpha_0 = 8\, \text{MW/Hz}, \beta_0 = 8.39\, \text{kVar/V}$. In the second scenario, we increase the droop coefficients to $\alpha_0 = 11\, \text{MW/Hz}, \beta_0 = 8.39\, \text{kVar/V}$. In this way, we force more initial power set-points to be outside the BESS feasible operation region.
\begin{figure}[!htbp]
	\vspace{-0.8cm}
	\hspace*{-0.4cm}
	\centering
	\includegraphics[width=1.02\linewidth]{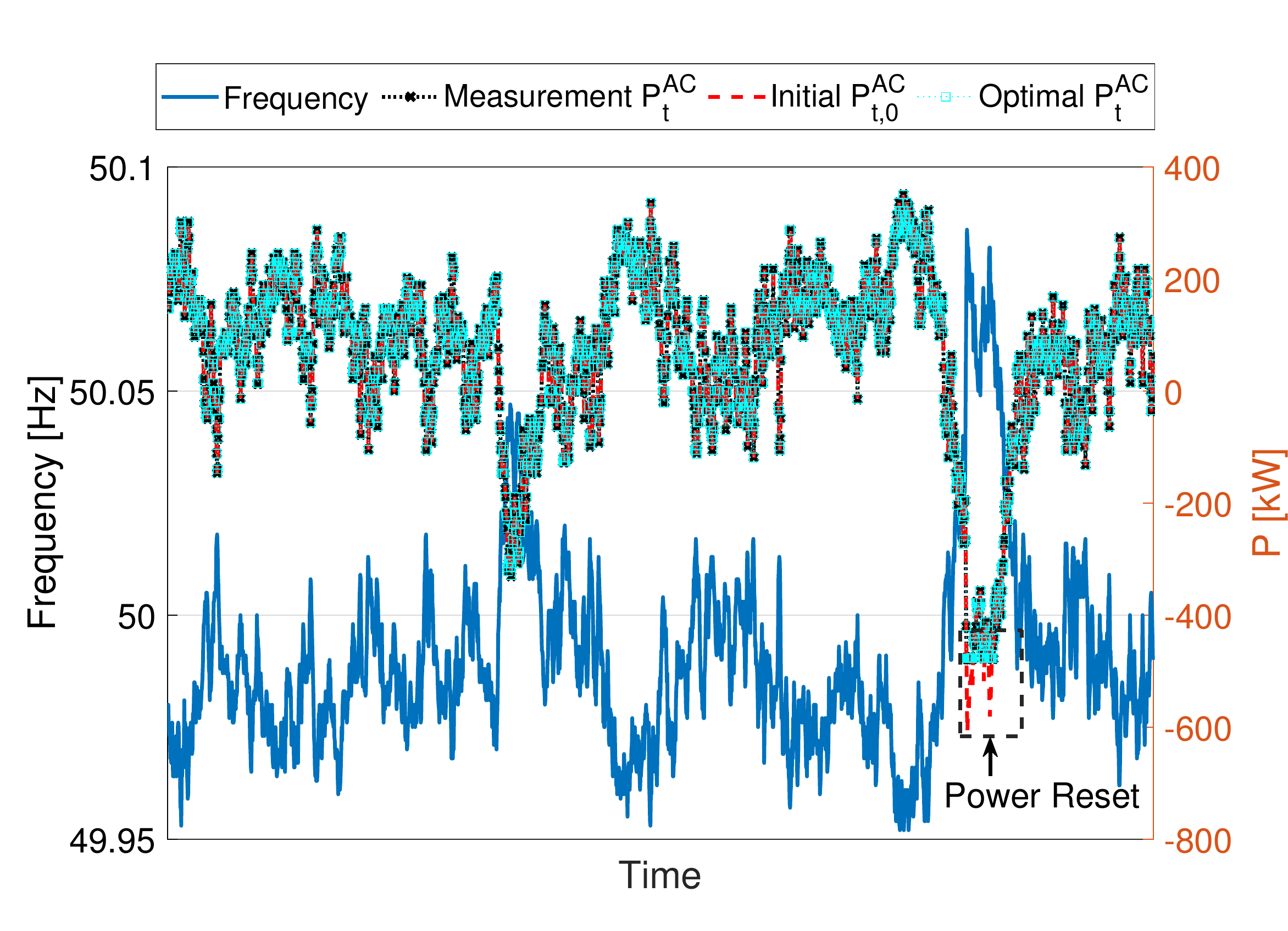}
	\vspace{-0.5cm}
	\caption{Algorithm \ref{alg:heuris} frequency control result for scenario 1.}
	\label{fig:frequency_response1}
\end{figure}
\begin{figure}[!htbp]
	\vspace{-0.5cm}
	\centering
	\hspace*{-0.4cm}
	\includegraphics[width=1.02\linewidth]{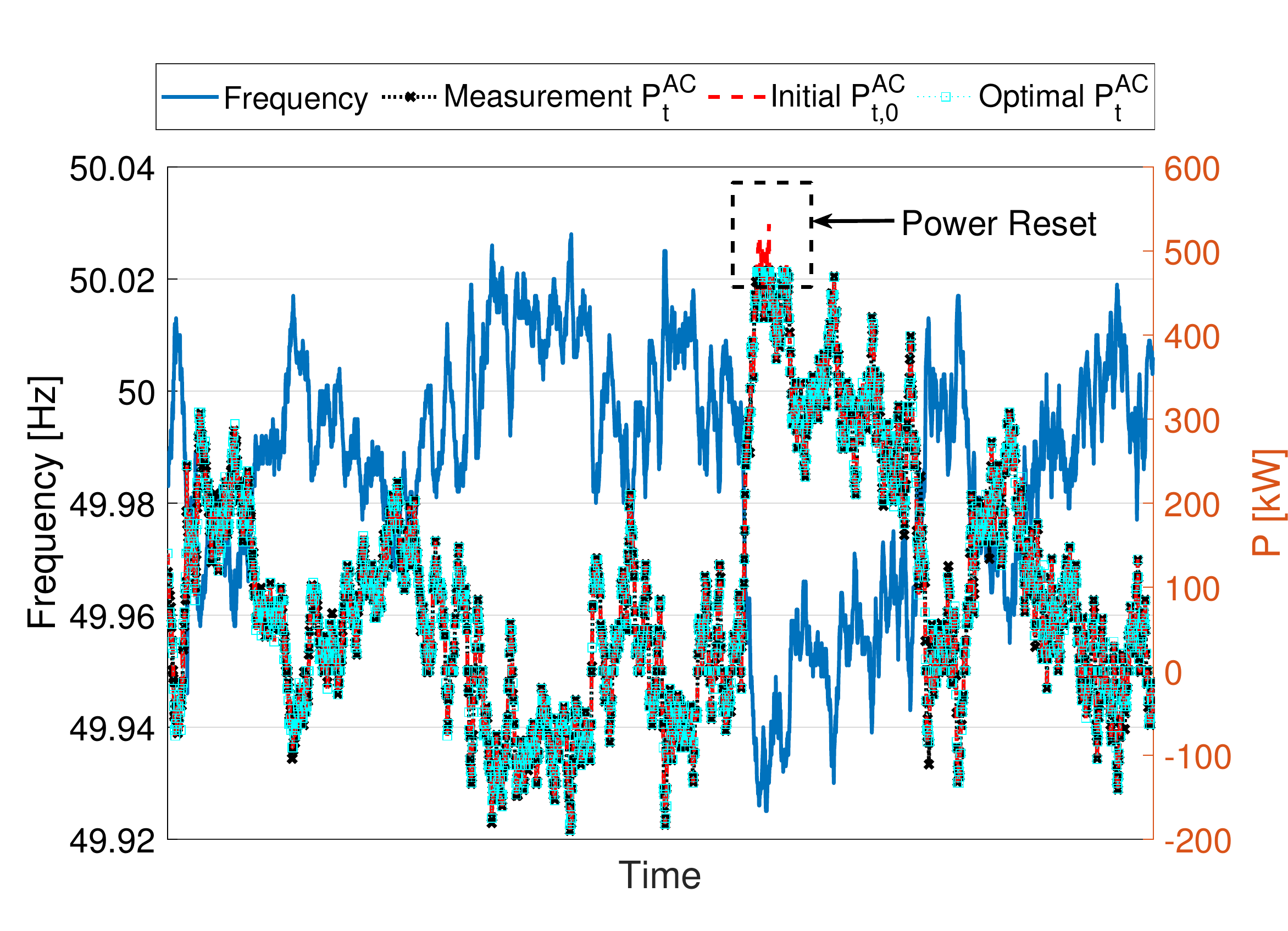}
	\vspace{-0.5cm}
	\caption{Algorithm \ref{alg:fast} frequency control result for scenario 1.}
	\label{fig:frequency_response2}
\end{figure}
\begin{figure}[!htbp]
	\vspace{-0.5cm}
	\centering
	\hspace*{-0.4cm}
	\includegraphics[width=1.02\linewidth]{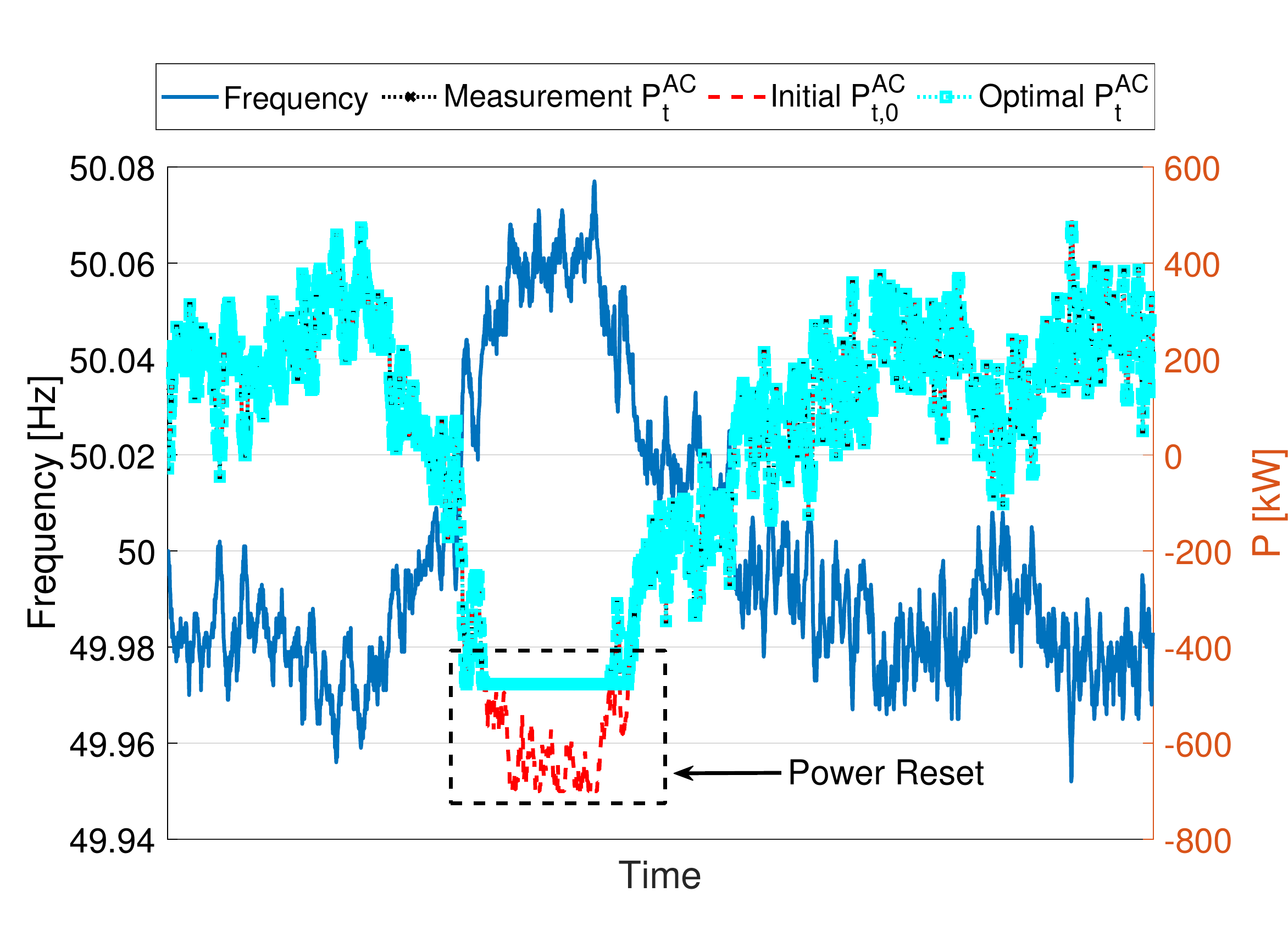}
	\vspace{-0.5cm}
	\caption{Algorithm \ref{alg:heuris} frequency control result for scenario 2.}
	\label{fig:frequency_response3}
\end{figure}
\begin{figure}[!htbp]
	\vspace{-0.5cm}
	\centering
	\hspace*{-0.4cm}
	\includegraphics[width=1.02\linewidth]{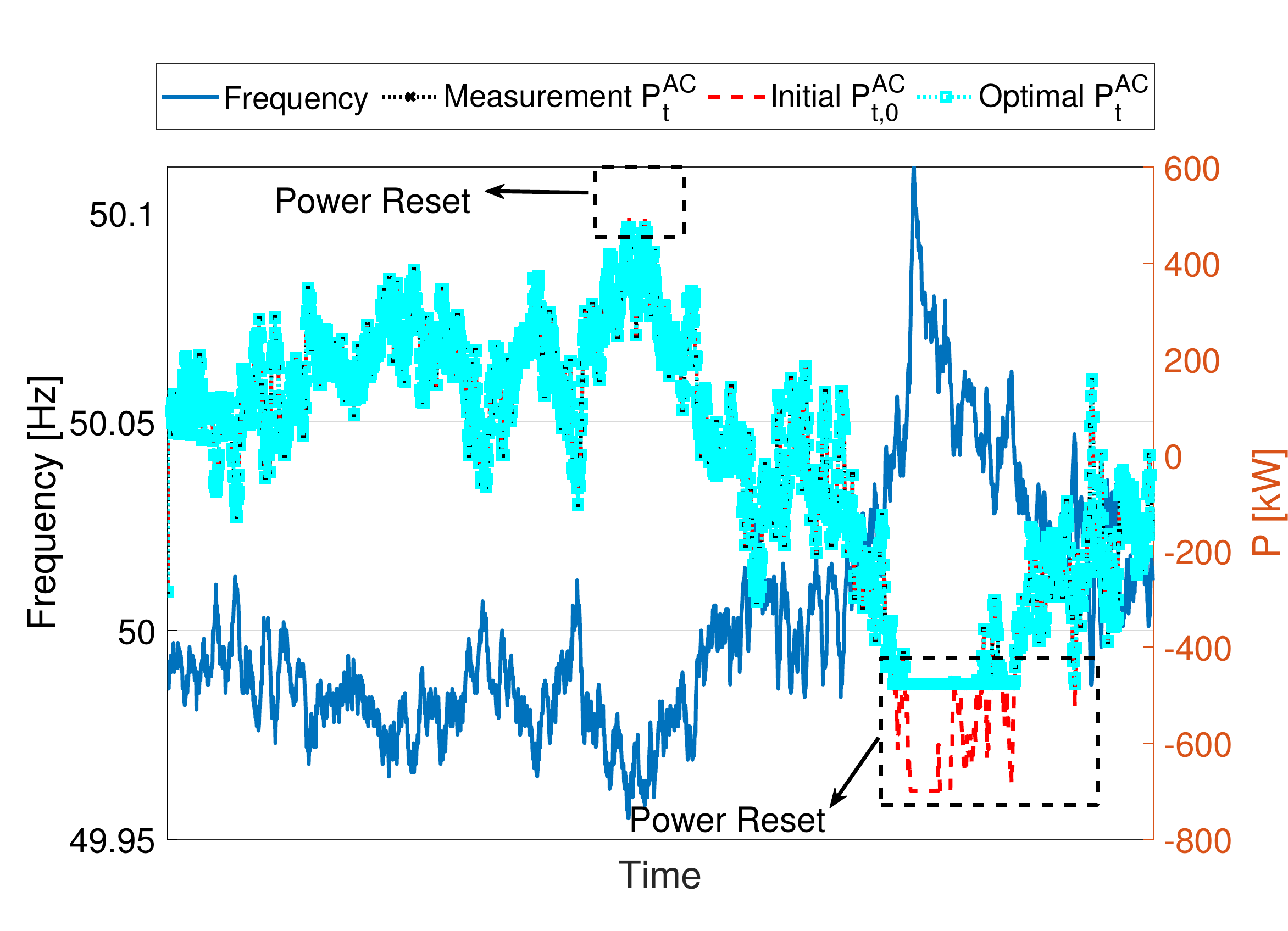}
	\vspace{-0.5cm}
	\caption{Algorithm \ref{alg:fast} frequency control result for scenario 2.}
	\label{fig:frequency_response4}
\end{figure}
\subsection{{Frequency Control}}
The frequency control results of Algorithm \ref{alg:heuris} and Algorithm \ref{alg:fast} are shown in Fig. \ref{fig:frequency_response1}, Fig. \ref{fig:frequency_response2}, Fig. \ref{fig:frequency_response3} and Fig. \ref{fig:frequency_response4}. Positive/negative values of the active power means discharging/charging of the BESS. In all these figures, we can observe the presence of some initial power set-points which are optimized effectively by our proposed algorithms. These results show that large grid frequency could happen any time during the operations and the frequency control services should be continuously provided. 
\subsection{{Voltage Support}}
The voltage support provided by using Algorithm \ref{alg:heuris} and Algorithm \ref{alg:fast} are shown in Fig. \ref{fig:voltage_support1}, Fig. \ref{fig:voltage_support2}, Fig. \ref{fig:voltage_support3} and Fig. \ref{fig:voltage_support4}. Similarly, positive/negative values of the reactive power means discharging/charging of the BESS. As it can be seen, there are some very large initial power set-points in Fig. \ref{fig:voltage_support4} which are optimized by Algorithm \ref{alg:fast}. The reason of a small number of large initial power set-points is due to the small value of the voltage droop coefficient that we are using in the experiments. Compared with voltage support, it is more important to provide frequency control services to the grid. From the converter capability curve we have shown before, the BESS ability to provide active power is also stronger than providing reactive power to the grid.
\begin{figure}[!htbp]
	\vspace{-1cm}
	\centering
	\hspace*{-0.4cm}
	\centering
	\includegraphics[width=1.02\linewidth]{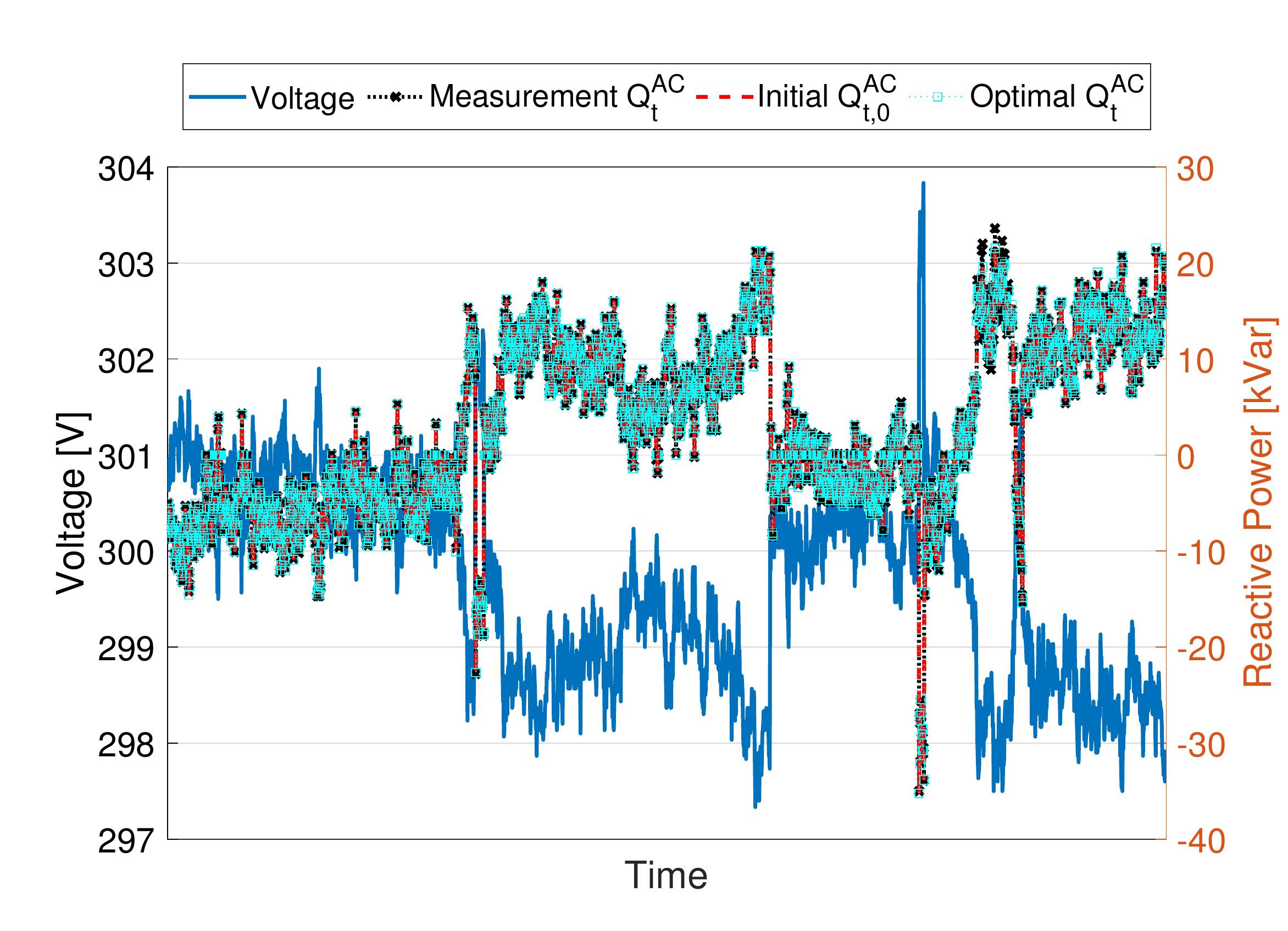}
	\vspace{-0.5cm}
	\caption{Algorithm \ref{alg:heuris} voltage control result for scenario 1.}
	\label{fig:voltage_support1}
\end{figure}
\begin{figure}[!htbp]
	\vspace{-0.8cm}
	\centering
	\hspace*{-0.4cm}
	\includegraphics[width=1.02\linewidth]{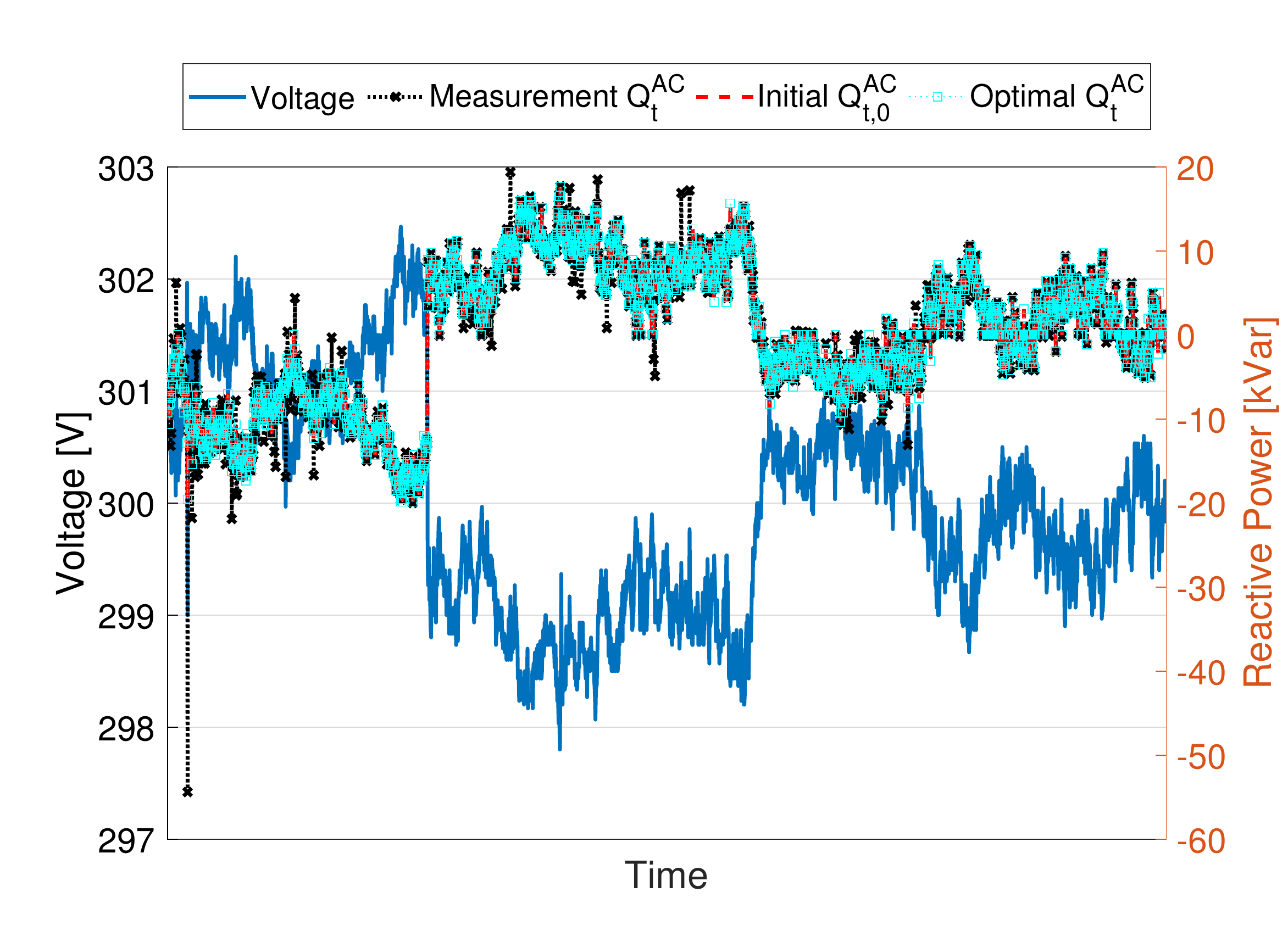}
	\vspace{-0.5cm}
	\caption{Algorithm \ref{alg:fast} voltage control result for scenario 1.}
	\label{fig:voltage_support2}
\end{figure}
\begin{figure}[!htbp]
	\vspace{-0.5cm}
	\hspace*{-0.4cm}
	\includegraphics[width=1.02\linewidth]{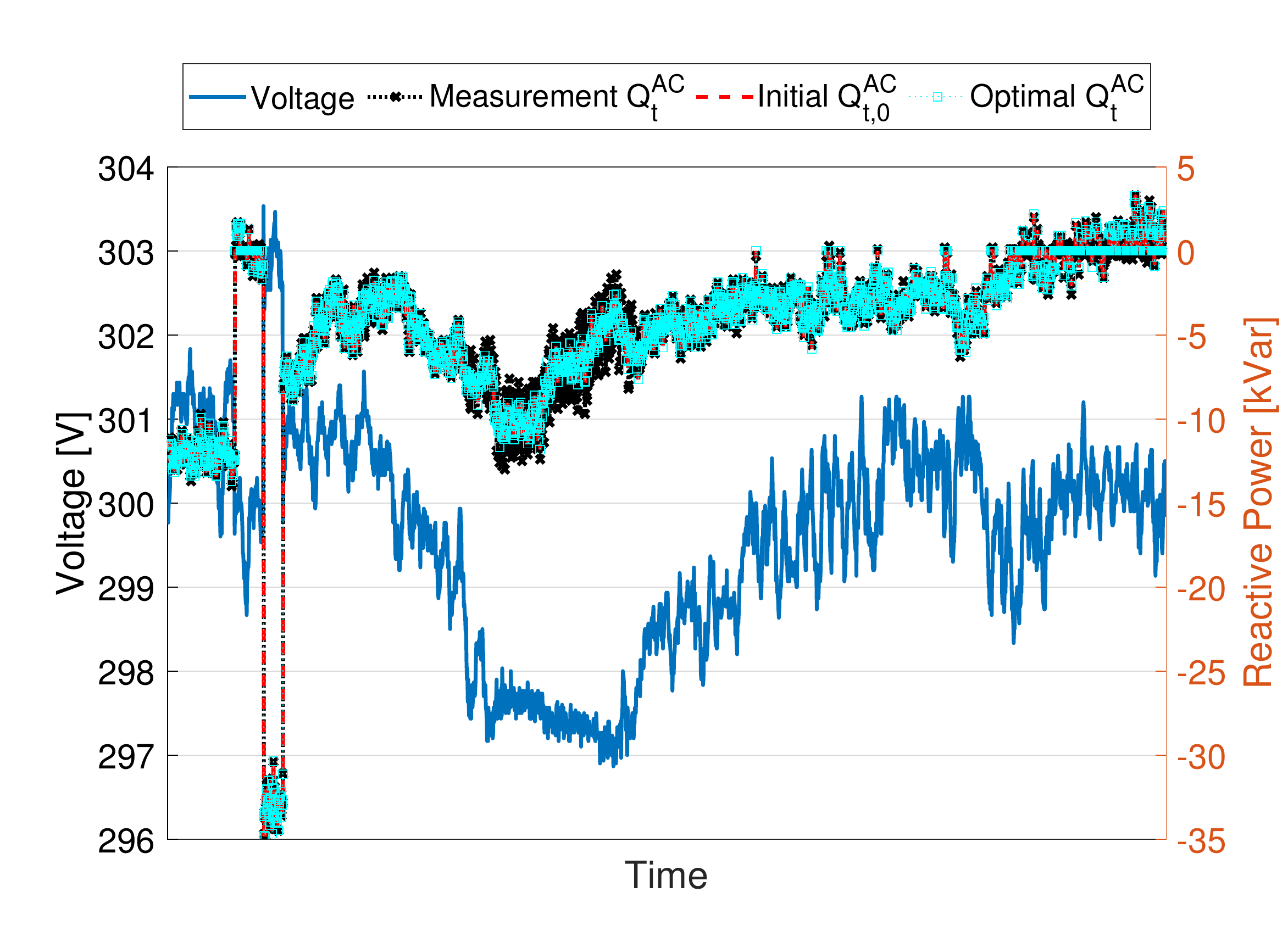}
	\vspace{-0.5cm}
	\caption{Algorithm \ref{alg:heuris} voltage control result for scenario 2.}
	\label{fig:voltage_support3}
\end{figure}
\begin{figure}[!htbp]
	\vspace{-0.5cm}
	\centering
	\hspace*{-0.4cm}
	\includegraphics[width=1.02\linewidth]{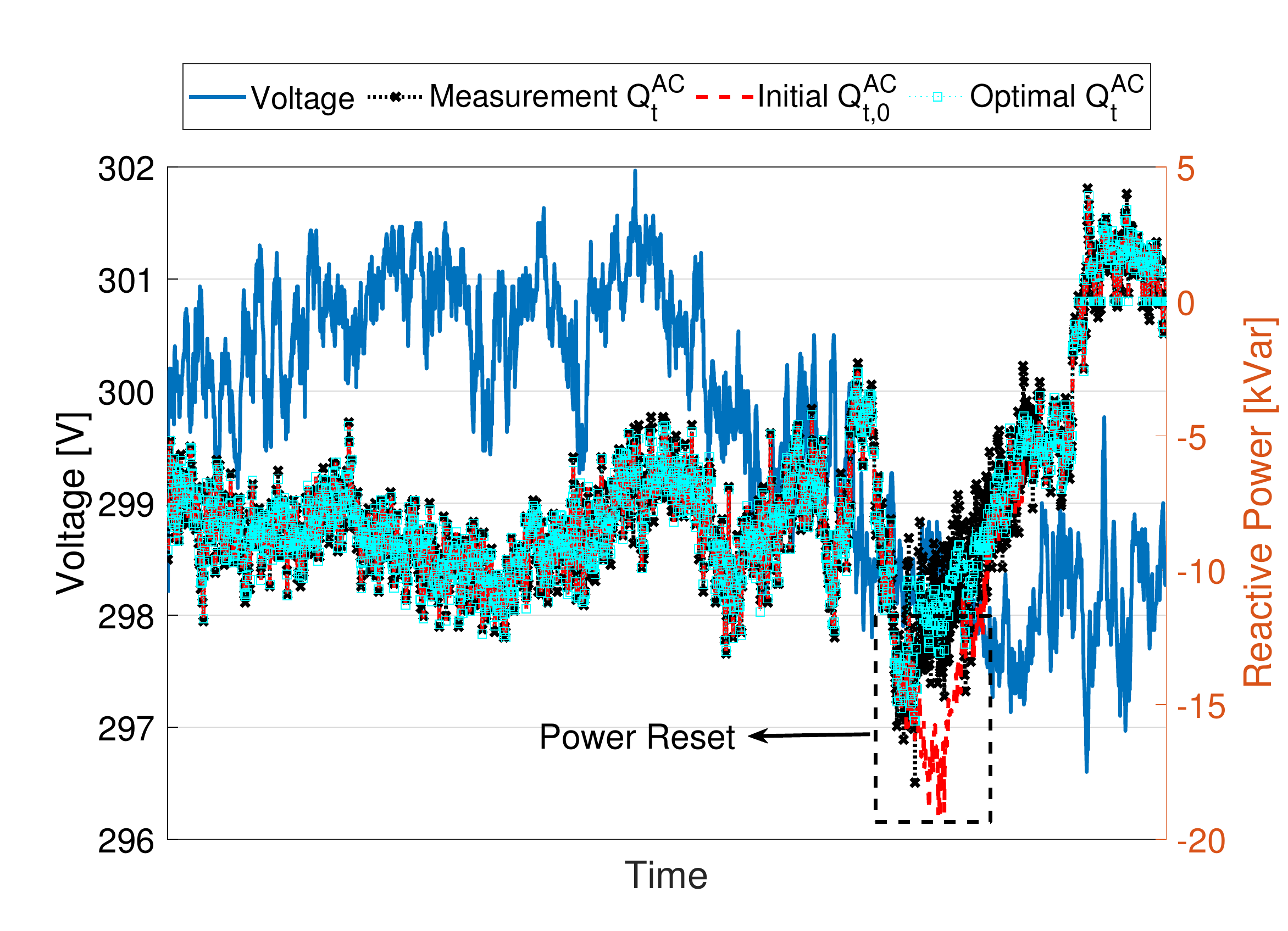}
	\vspace{-0.5cm}
	\caption{Algorithm \ref{alg:fast} voltage control result for senario 2.}
	\label{fig:voltage_support4}
\end{figure}
\subsection{Computational Efficiency}
The computation time results shown in Fig. \ref{fig:time1}, Fig. \ref{fig:time2}, Fig. \ref{fig:time3} and Fig. \ref{fig:time4} are the histogram plot of the instances of the BESS real-time control loop. We can see that, for most cases, the computation time of Algorithm \ref{alg:heuris} is around 80 ms. The maximum computation time is always less than 120 ms which gives sufficient time to update our control-loop within 200 ms. For Algorithm \ref{alg:fast}, there is a huge improvement compared to Algorithm \ref{alg:heuris}. The computation time is around 7 ms. The maximum computation time of Algorithm \ref{alg:fast} is always less than 25 ms which gives sufficient time to update our control-loop within 100 ms. This achievement is essential for the BESS real-time control.
\begin{figure}[!htbp]
	\vspace{-0.5cm}
	\centering
	\includegraphics[width=0.95\linewidth]{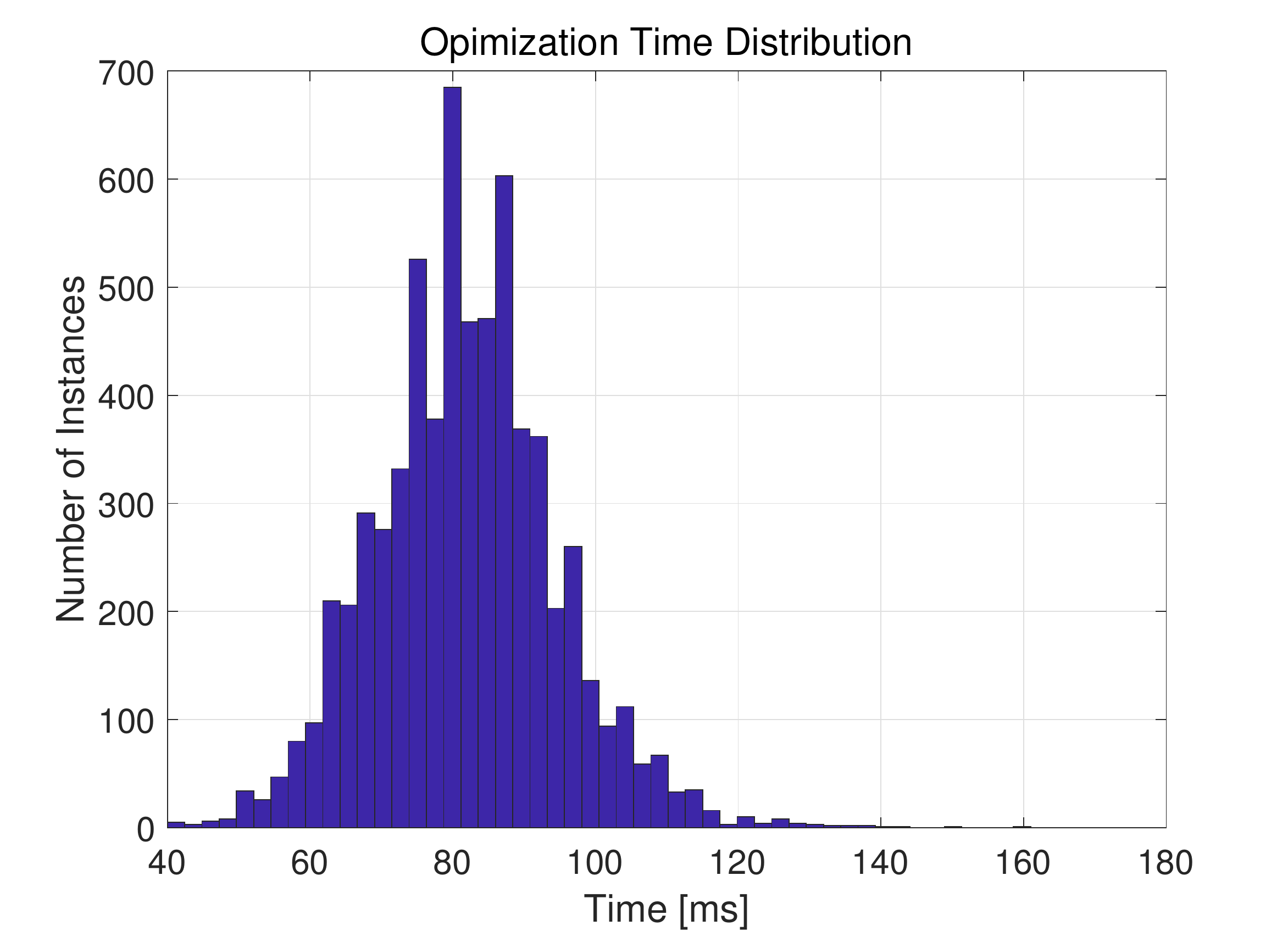}
	\vspace{-0.2cm}
	\caption{Algorithm \ref{alg:heuris} computational efficiency for scenario 1.}
	\label{fig:time1}
\end{figure}
\begin{figure}[!htbp]
	\vspace{-0.4cm}
	\centering
	\includegraphics[width=0.95\linewidth]{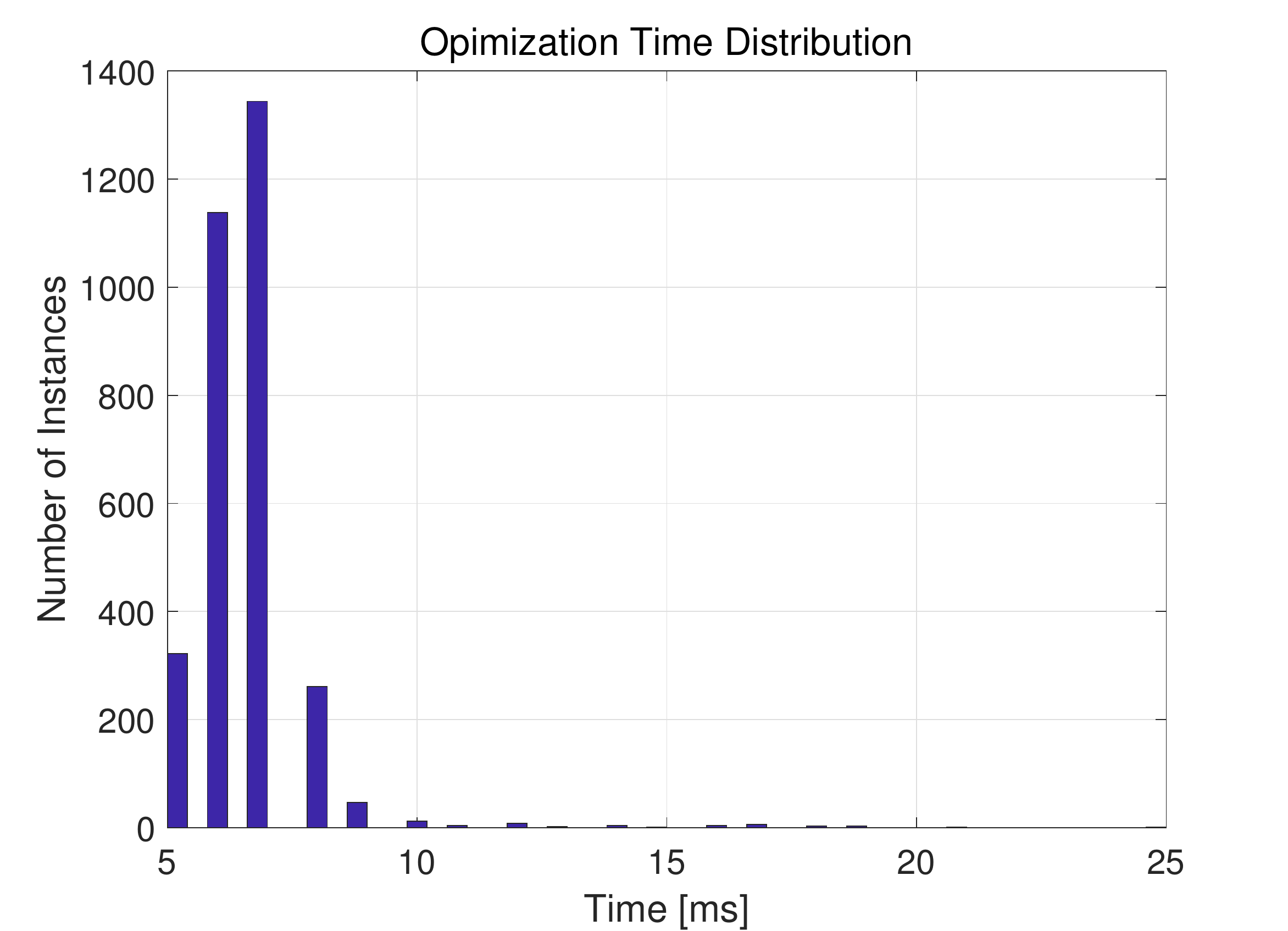}
	\vspace{-0.2cm}
	\caption{Algorithm \ref{alg:fast} computational efficiency for scenario 1.}
	\label{fig:time2}
\end{figure}
\begin{figure}[!htbp]
	\centering
	\vspace{-0.4cm}
	\includegraphics[width=0.95\linewidth]{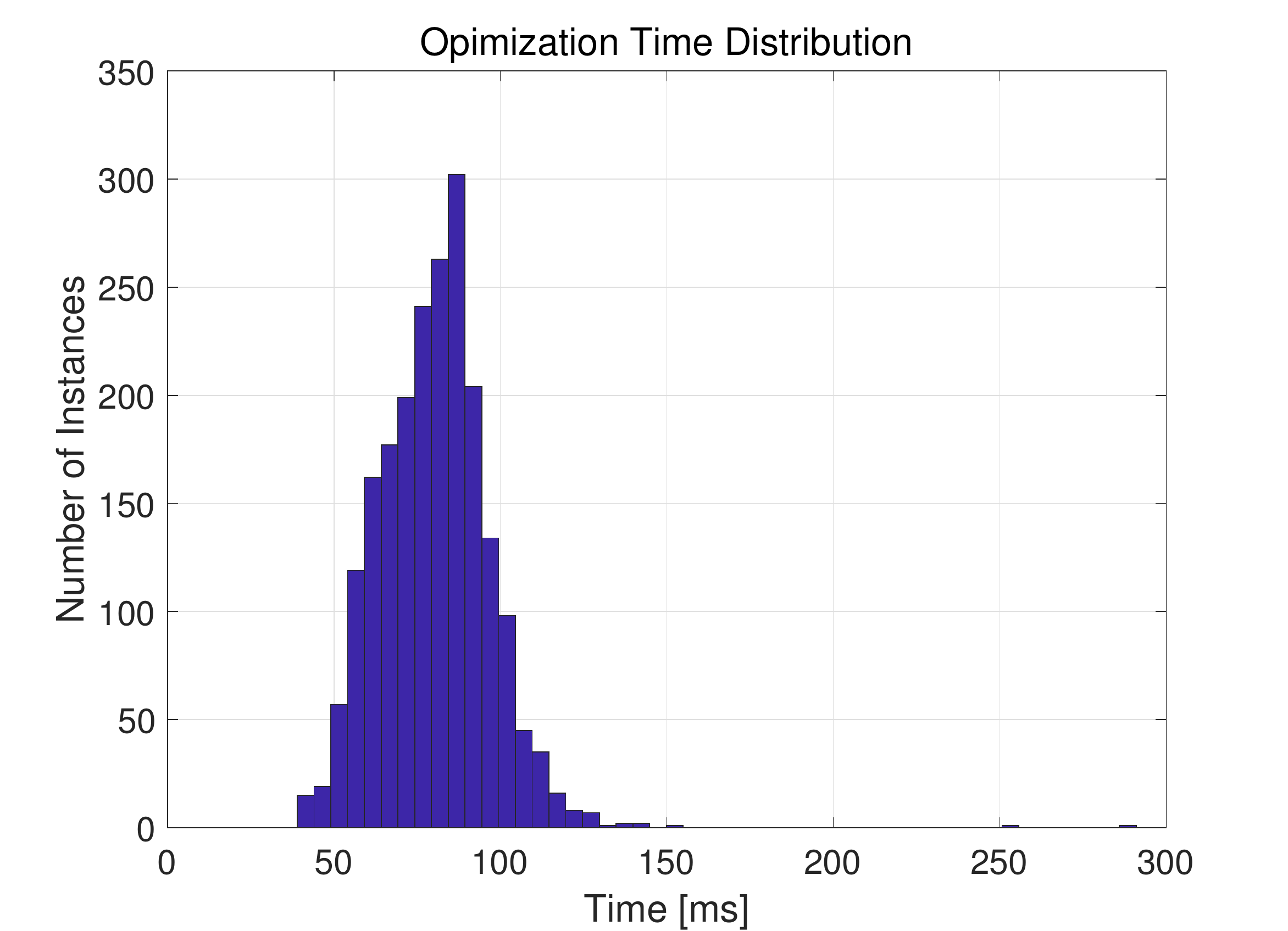}
	\vspace{-0.2cm}
	\caption{Algorithm \ref{alg:heuris} computational efficiency for scenario 2.}
	\label{fig:time3}
\end{figure}
\begin{figure}[!htbp]
	\vspace{-0.5cm}
	\centering
	\includegraphics[width=0.95\linewidth]{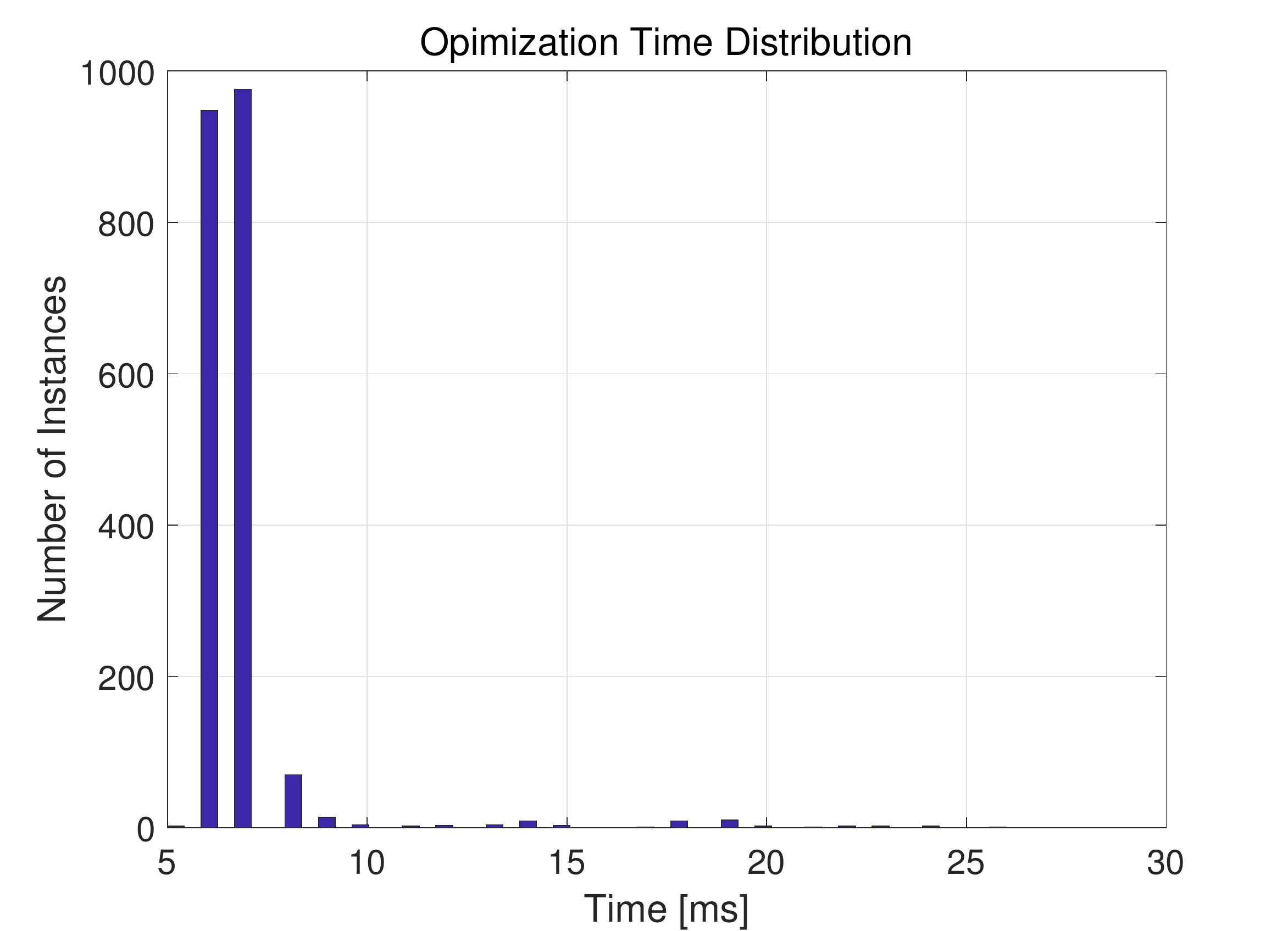}
	\vspace{-0.2cm}
	\caption{Algorithm \ref{alg:fast} computational efficiency for scenario 2.}
	\label{fig:time4}
\end{figure}
\subsection{Energy}
We define the following metrics to quantify the performance of the proposed BESS real-time control in terms of the provided energy to the grid.\\
Since the battery is discharging when $P^{AC}_{t}>0$, the total discharged energy (TDE) is defined as:
\begin{subequations}
\begin{align}
	\text{TDE}=\sum_{t\in \left\{ t | P^{AC}_{t}>0 \right \}}P^{AC}_{t}\Delta t \label{eq:tde}
\end{align}
Since battery is charging when $P^{AC}_{t}<0$, the total charged energy (TCE) is defined as:
\begin{align}
	\text{TCE}=\sum_{t\in \left\{ t | P^{AC}_{t}<0 \right \}}|P^{AC}_{t}|\Delta t \label{eq:tce}
\end{align}
Since both Algorithm 1 and Algorithm 2 can guarantee the continuous operation of the BESS when the initial power set-point is not feasible, we define the provided energy during these time intervals as the total sustained energy (TSE) i.e. the energy which cannot be provided without using Algorithm \ref{alg:heuris} or Algorithm \ref{alg:fast}. In other words, if Algorithm \ref{alg:heuris} or Algorithm \ref{alg:fast} is not used, the converter will automatically re-set the power set-points to zero when the initial power set-points are outside the feasible region of the capability curve. The TSE is defined as:
\begin{align}
	\text{TSE}=\sum_{t \in \left\{ t | P^{AC}_{0,t} \notin (1c) \right \}} |P^{AC}_{t}|\Delta t \label{eq:tse}
\end{align}
\end{subequations}
We summarize the {values of these metrics} of both scenarios in Table \ref{tab:energy}. More TSE are provided in scenario 2. This is reasonable since larger droop coefficients force more initial power set-points to be outside of the feasible region of the converter capability curve. By using our proposed algorithms, we an guarantee the continuous operation of the BESS.
\begin{table}[!htbp]
  \centering
  \caption{Charged/Discharged Energy and Reused Energy}
    \begin{tabular}{|c|c|c|c|c|}\hline
    Droop Coefficient & \multicolumn{2}{c|}{8 MWh/Hz} & \multicolumn{2}{c|}{11 MWh/Hz} \\ \hline
    Algorithm & 1     & 2     & 1     & 2 \\ \hline
    TDE [kWh] & 65.77 & 114.52 & 127.78 & 106.70 \\\hline
    TCE [kWh] & 55.96 & 18.36 & 102.11 & 110.16 \\\hline
    TSE [kWh] & 2.08  & 4.17  & 58.95 & 48.09 \\\hline
    \end{tabular}%
  \label{tab:energy}%
  \vspace{-0.4cm}
\end{table}%
\section{Conclusion}
To provide reliable frequency control and voltage support from BESS, we formulate the real-time control problem in a nonlinear optimization model PQ-opt-o taking into account the {voltage-dependent} DC-AC converter capability and battery security constraints. We propose, and {rigorously} prove to use the convex optimization model PQ-opt-m, and solution algorithm \ref{alg:heuris} to find the global optimal power set-points of the original optimization model PQ-opt-o. To improve the computational performance of PQ-opt-m, we propose a fast real-time control algorithm \ref{alg:fast} by approximating PQ-opt-m and discretizing the feasible region of the optimization model. Our optimization model and solution algorithms are proved analytically and validated experimentally. The experimental results show that we can reach a time latency of 200 ms to update the real-time control loop by using algorithm \ref{alg:heuris} which also gives accurate optimal power set-points solutions. Algorithm \ref{alg:fast} achieves 100 ms of time latency to update the real-time control loop. This algorithm avoids the usage of optimization solver in the real-time control of BESS though sacrifices a bit the solution accuracy. The proposed real-time control improves the efficiency, security and continuity of the BESS operations in providing frequency control and voltage support services to the grid.    
\bibliographystyle{IEEEtran}
\bibliography{biblio}{}
\begin{IEEEbiography}[{\includegraphics[width=1in,height=1.25in]{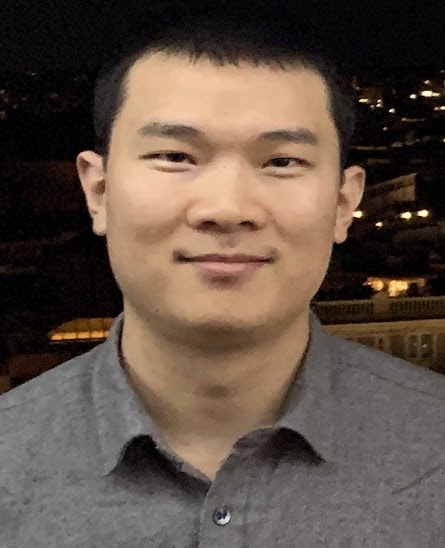}}]{Zhao Yuan}
received joint PhD degree from KTH Royal Institute of Technology, Comillas Pontifical University and Delft University of Technology in 2018. He worked as a scientist at the Swiss Federal Institute of Technology from 2019 to 2021. In 2010, he was awarded the National Second Prize in China Contemporary Undergraduate Mathematical Contest in Modeling by China Society for Industrial and Applied Mathematics. Zhao proposed and proved the thereoms about the existence and uniqueness of the optimal solution of the convex optimal power flow model based on second-order cone programming. He developed the Energy Management System (EMS) of the 560kWh/720kVA Battery Energy Storage System (BESS) on EPFL campus. Zhao co-developed the Smart Grid in Aigle Switzerland in 2020.
His research work on coordinated transimission-distribution operation was awarded by Advances in Engineering (AIE) as Key Scientific Article Contributing to Excellence in Science and Engineering Research in 2017. Zhao's article “Second-order cone AC optimal power flow: convex relaxations and feasible solutions” was awarded as the 2019 Best Paper in the Journal of Modern Power Systems and Clean Energy (MPCE). Webiste: https://sites.google.com/view/yuanzhao/
\end{IEEEbiography}
\begin{IEEEbiography}{Antonio Zecchino}
is a postdoctoral researcher at the Swiss Federal Institute of Technology, Lausanne, Switzerland. He received PhD degree from Technical University of Denmark in 2019. He received the B.Sc. degree in energy engineering and the M.Sc. degree in electrical engineering from the University of Padua, Padua, in 2012 and 2015, respectively.
\end{IEEEbiography}
\begin{IEEEbiography}{Rachid Cherkaoui}
is a senior scientist at the Swiss Federal Institute of Technology (EPFL), Lausanne, Switzerland. Rachid Cherkaoui received the M.Sc. and Ph.D. degrees in electrical engineering in 1983 and 1992, respectively, from EPFL. 
\end{IEEEbiography}
\begin{IEEEbiography}{Mario Paolone} (M’07–SM’10) received the M.Sc. (Hons.) and Ph.D. degrees in electrical engineering from the University of Bologna, Italy, in 1998 and 2002. In 2005, he was an Assistant Professor in power systems with the University of Bologna, where he was with the Power Systems Laboratory until 2011. Since 2011, he has been with the Swiss Federal Institute of Technology, Lausanne, Switzerland, where he is currently a Full Professor and the Chair of the Distributed Electrical Systems Laboratory. His research interests focus on power systems with particular reference to real-time monitoring and operational aspects, power system protections, dynamics and transients. Dr. Paolone has authored or co-authored over 300 papers published in mainstream journals and international conferences in the area of energy and power systems that received numerous awards including the IEEE EMC Technical Achievement Award, two IEEE Transactions on EMC best paper awards, the IEEE Power System Dynamic Performance Committee’s prize paper award and the Basil Papadias best paper award at the 2013 IEEE PowerTech. Dr. Paolone was the founder Editor-in-Chief of the Elsevier journal Sustainable Energy, Grids and Networks.
\end{IEEEbiography}
\vfill
\end{document}